\documentclass[onecolumn]{IEEEtran}     

\IEEEoverridecommandlockouts
\usepackage{setspace}
\usepackage[T1]{fontenc}

\usepackage{amssymb, mathtools, eucal}
\usepackage{amsthm}
\usepackage{rotating}

\theoremstyle{definition}
\newtheorem{definition}{Definition}
\newtheorem{theorem}{Theorem}
\newtheorem{proposition}[theorem]{Proposition}
\newtheorem{lemma}[theorem]{Lemma}
\newtheorem{corollary}[theorem]{Corollary}
\newtheorem{remark}{Remark}
\newtheorem{example}{Example}
\newtheorem{construction}{Construction}
\newtheorem{conjecture}{Conjecture}

\usepackage{graphicx}
\usepackage[dvipsnames]{xcolor}
\usepackage{color}
\usepackage{tikz}
\usepackage{afterpage}

\usepackage{multirow}  
\usetikzlibrary{calc}

\usepackage{array}    
\newcolumntype{C}[1]{>{\centering\arraybackslash}m{#1}}
\usepackage{arydshln}
\usepackage{multicol}
\usepackage{enumerate}

\usepackage{cite}
\usepackage{url}
\usepackage{hyperref}
\newcommand\myshade{70}
\hypersetup{
  linkcolor  = red!\myshade!black,
  citecolor  = blue!\myshade!black,
  urlcolor   = blue!\myshade!black,
  colorlinks = true,
}

\usepackage{algorithm}
\usepackage{algpseudocode}
\algnewcommand\algorithmicinput{\textbf{Input:}}
\algnewcommand\INPUT{\item[\algorithmicinput]}
\algnewcommand\algorithmicoutput{\textbf{Output:}}
\algnewcommand\OUTPUT{\item[\algorithmicoutput]}

\usepackage{comment}
\usepackage[normalem]{ulem}

\newcommand{\hanmao}[1]{{\footnotesize [\gcomment{#1}\;\;\rcomment{--HM}]}}

\newcommand{\wenqin}[1]{{\footnotesize [\textcolor{cyan!70!black}{#1} \textcolor{blue!80!black}{--WQ}]}}

\newcommand{\acomment}[2]{{\color{#1}#2}}
\newcommand{\rcomment}[1]{\acomment{red}{#1}}

\newcommand{\gcomment}[1]{\acomment{OliveGreen}{#1}}

\newcommand\nc\newcommand

\nc{\vzero}{{\boldsymbol{0}}}
\nc{\vone}{{\boldsymbol{1}}}

\nc{\bfC}{{\boldsymbol C}}
\nc{\Read}{{\sf R}}
\nc{\cost}{{\sf cost}}
\nc{\bfa}{{\boldsymbol a}}  \nc{\bfb}{{\boldsymbol b}}  \nc{\bfd}{{\boldsymbol d}}
\nc{\bfe}{{\boldsymbol e}}  \nc{\bff}{{\boldsymbol f}}  \nc{\bfg}{{\boldsymbol g}}
\nc{\bfh}{{\boldsymbol h}}  \nc{\bfi}{{\boldsymbol i}}  \nc{\bfj}{{\boldsymbol j}}
\nc{\bfk}{{\boldsymbol k}}  \nc{\bfl}{{\boldsymbol l}} 
 \nc{\bfm}{{\boldsymbol m}}
\nc{\bfn}{{\boldsymbol n}}  \nc{\bfo}{{\boldsymbol o}}  \nc{\bfp}{{\boldsymbol p}}
\nc{\bfq}{{\boldsymbol q}}  \nc{\bfr}{{\boldsymbol r}}  \nc{\bfs}{{\boldsymbol s}}
\nc{\bft}{{\boldsymbol t}}  \nc{\bfu}{{\boldsymbol u}}  \nc{\bfv}{{\boldsymbol v}}
\nc{\bfw}{{\boldsymbol w}}  \nc{\bfx}{{\boldsymbol x}}  \nc{\bfy}{{\boldsymbol y}}
\nc{\bfz}{{\boldsymbol z}}

\nc\bfA{{\boldsymbol A}}  \nc\bfB{{\boldsymbol B}}  \nc\bfD{{\boldsymbol D}}
\nc\bfE{{\boldsymbol E}}  \nc\bfF{{\boldsymbol F}}  \nc\bfG{{\boldsymbol G}}
\nc\bfH{{\boldsymbol H}}  \nc\bfI{{\boldsymbol I}}  \nc\bfJ{{\boldsymbol J}}
\nc\bfK{{\boldsymbol K}}  \nc\bfL{{\boldsymbol L}}  \nc\bfM{{\boldsymbol M}}
\nc\bfN{{\boldsymbol N}}  \nc\bfO{{\boldsymbol O}}  \nc\bfP{{\boldsymbol P}}
\nc\bfQ{{\boldsymbol Q}}  \nc\bfR{{\boldsymbol R}}  \nc\bfS{{\boldsymbol S}}
\nc\bfT{{\boldsymbol T}}  \nc\bfU{{\boldsymbol U}}  \nc\bfV{{\boldsymbol V}}
\nc\bfW{{\boldsymbol W}}  \nc\bfX{{\boldsymbol X}}  \nc\bfY{{\boldsymbol Y}}
\nc\bfZ{{\boldsymbol Z}}
\nc{\bfc}{{\boldsymbol c}}

\nc\cA{\mathcal{A}}  \nc\cB{\mathcal{B}}  \nc\cC{\mathcal{C}}
\nc\cD{\mathcal{D}}  \nc\cE{\mathcal{E}}  \nc\cF{\mathcal{F}}
\nc\cG{\mathcal{G}}  \nc\cH{\mathcal{H}}  \nc\cI{\mathcal{I}}
\nc\cJ{\mathcal{J}}  \nc\cK{\mathcal{K}}  \nc\cL{\mathcal{L}}
\nc\cM{\mathcal{M}}  \nc\cN{\mathcal{N}}  \nc\cO{\mathcal{O}}
\nc\cP{\mathcal{P}}  \nc\cQ{\mathcal{Q}}  \nc\cR{\mathcal{R}}
\nc\cS{\mathcal{S}}  \nc\cT{\mathcal{T}}  \nc\cU{\mathcal{U}}
\nc\cV{\mathcal{V}}  \nc\cW{\mathcal{W}}  \nc\cX{\mathcal{X}}
\nc\cY{\mathcal{Y}}  \nc\cZ{\mathcal{Z}}

\nc\bbA{\mathbb{A}}  \nc\bbB{\mathbb{B}}  \nc\bbC{\mathbb{C}}
\nc\bbD{\mathbb{D}}  \nc\bbE{\mathbb{E}}  \nc\bbF{\mathbb{F}}
\nc\bbG{\mathbb{G}}  \nc\bbH{\mathbb{H}}  \nc\bbI{\mathbb{I}}
\nc\bbJ{\mathbb{J}}  \nc\bbK{\mathbb{K}}  \nc\bbL{\mathbb{L}}
\nc\bbM{\mathbb{M}}  \nc\bbN{\mathbb{N}}  \nc\bbO{\mathbb{O}}
\nc\bbP{\mathbb{P}}  \nc\bbQ{\mathbb{Q}}  \nc\bbR{\mathbb{R}}
\nc\bbS{\mathbb{S}}  \nc\bbT{\mathbb{T}}  \nc\bbU{\mathbb{U}}
\nc\bbV{\mathbb{V}}  \nc\bbW{\mathbb{W}}  \nc\bbX{\mathbb{X}}
\nc\bbY{\mathbb{Y}}  \nc\bbZ{\mathbb{Z}}
\newcommand{\bigboxplus}{\mathop{\scalebox{1.5}{$\boxplus$}}\limits}

\nc\norm[1]{\left\lVert#1\right\rVert}  
\nc\setof[1]{\left\{ #1 \right\}}  

\nc\ceil[1]{\left\lceil#1\right\rceil}              
\nc\floor[1]{\left\lfloor#1\right\rfloor}           

\nc\entropy{{\sf H}}        
\nc\dist{{\rm d}}           
\nc\wt{{\sf wt}}            

\title{Zigzag Codes Revisited: From Optimal Rebuilding to Small Skip Cost and Small Fields}
\author{

  \IEEEauthorblockN{ Wenqin Zhang\IEEEauthorrefmark{1},
  Han Mao Kiah\IEEEauthorrefmark{2}, 
  Son Hoang Dau\IEEEauthorrefmark{3}
 }\\
   \IEEEauthorblockA{
    \small \IEEEauthorrefmark{1}School of Cyber Science and Technology, Hubei University, Wuhan, China
    }\\
\IEEEauthorblockA{\small 
 \IEEEauthorrefmark{2}School of Physical and Mathematical Sciences, Nanyang Technological University, Singapore
    }\\
\IEEEauthorblockA{\small 
 \IEEEauthorrefmark{3}School of Computing Technologies, RMIT University, Melbourne, VIC, Australia
    }\\
  {\footnotesize wenqin\_zhang@hubu.edu.cn,
  hmkiah@ntu.edu.sg, sonhoang.dau@rmit.edu.au}
  \vspace{-5mm}
  \thanks{\IEEEauthorrefmark{5}Corresponding author: Wenqin Zhang.}
  }
\date{}

\begin{document}

\maketitle
\begin{abstract}
We revisit zigzag array codes, a family of MDS codes known for achieving optimal access and optimal rebuilding ratio in single-node repair. 
In this work, we endow zigzag codes with two new properties: small field size and low skip cost. 

First, we prove that when the row-indexing group is $\cG=\bbZ_2^m$ and the field has characteristic two, explicit coefficients over any field with $|\bbF|\ge N$ guarantee the MDS property, thereby decoupling the dependence among $p$, $k$, and $M$. 
Second, we introduce an ordering-and-subgroup framework that yields repair-by-transfer schemes with bounded skip cost and low repair-fragmentation ratio (RFR), while preserving optimal access and optimal rebuilding ratio. 
Our explicit constructions include families with zero skip cost whose rates approach $2/3$, and families with bounded skip cost whose rates approach $3/4$ and $4/5$. 
These rates are comparable to those of MDS array codes widely deployed in practice. 
Together, these results demonstrate that zigzag codes can be made both more flexible in theory and more practical for modern distributed storage systems.
\end{abstract}

\section{Introduction}\label{sec:intro}

Large-scale distributed storage systems (DSS), such as the Google File System and the Hadoop Distributed File System, require strong fault tolerance to protect data against node failures. 
Among various redundancy schemes, \emph{maximum distance separable} (MDS) codes are widely used because they provide the maximum failure tolerance for a given storage overhead. 
An $(N,k)$ MDS code encodes $k$ data blocks into $N$ fragments so that the original data can be recovered from any $k$ of them. 
However, repairing a single failed node with an MDS code is inefficient: it still requires downloading the entire data from $k$ other nodes, leading to high \emph{repair bandwidth}, i.e., the total amount of data transferred from \emph{helper nodes} during repair.  
To address this inefficiency, Dimakis \textit{et al.}~\cite{dimakis2010network} proposed the study of repair bandwidth for array codes and established a fundamental lower bound. 
An $(N,k,M)$ \emph{MDS array code} is an $M \times N$ array over a finite field $\mathbb{F}_q$, where $M$ is the \emph{subpacketization level}, $N$ is the total number of nodes with $k$ systematic nodes and $p=N-k$ parity nodes, such that any $k$ out of $N$ nodes suffice to recover the file. 
A particularly important subclass is the \emph{minimum storage regenerating} (MSR) codes, which simultaneously satisfy the MDS property and meet this lower bound on repair bandwidth. 
Over the past decade, extensive research has focused on constructing MDS array codes with optimal repair bandwidth~\cite{wang2011codes,Tamo2013,papailiopoulos2013repair,sasidharan2015high,wang2016explicit,ye2017explicit,liu2018explicit,li2020systematic}.

While minimizing repair bandwidth is crucial, it does not fully capture the repair cost in practice. 
Disks in real-world storage systems have limited computational ability, so MDS array codes with \emph{uncoded repair schemes}, in which helper nodes transmit raw symbols without computation, are more desirable. 
When repairing a failed node, if the amount of data \emph{accessed} from the helper nodes also meets the cut-set bound, the code is said to have the \emph{optimal access property}~\cite{ye2017explicit,vajha2018clay}.

Even with optimal repair bandwidth and optimal access, practical systems may still suffer from significant latency. 
Specifically, \cite{wu2021} attributed access latency at helper nodes to two factors: data processing time and data access time. 
The former was addressed in earlier works where the concept of \emph{repair-by-transfer} was introduced~\cite{shah2012}. In the repair-by-transfer framework, a failed node is repaired by a simple transfer of stored symbols without any computation at the helpers%
\footnote{In~\cite{shah2012}, an even stricter model was considered, where the failed node itself performs no computation. In this work, we only assume that the helpers perform no computation.}.
The latter (data access time), however, depends heavily on how data is laid out on disk. Two strategies are particularly effective in reducing access latency:  
(i) designing codes with \emph{small subpacketization} levels to reduce fragmentation, and  
(ii) ensuring that the symbols read during repair are located in \emph{contiguous sections}, so that random I/O overhead is avoided. The non-contiguity of these reads is measured by the \emph{skip cost}.  
Fig.~\ref{fig:zigzag} shows the skip cost for two array codes. 
 In (a), we have five nodes $\bfa^{(0)}$, $\bfa^{(1)}$, $\bfa^{(2)}$, $\bfp^{(0)}$ and $\bfp^{(1)}$. 
When $\bfa^{(2)}$ fails, the contents highlighted in {\color{blue} blue} are read. Notably, in node $\bfa^{(0)}$, the read contents are not contiguous and we quantify the gap as skip cost one. 
Summing up the skip costs across all helper nodes, we say that the skip cost of repairing node $\bfa^{(2)}$  is four.
In contrast, (b) illustrates an example of our proposed code construction. In this case, the repair bandwidth remains the same, but we see that all read contents are contiguous. Consequently, the skip cost is zero.

These considerations motivate the need for additional performance metrics beyond bandwidth and access volume. 
In this work we introduce, and formalize later in Section~II-C, the \emph{repair-fragmentation ratio} (RFR): a normalized metric that jointly accounts for the total amount of data read and its contiguity across helpers. 
RFR enables a fair comparison between codes that achieve the same repair bandwidth but exhibit very different I/O behaviors in practice.

\begin{figure}[H]
    \footnotesize
    \begin{center}
        \noindent(a) $(5, 3, 4)$-array code with skip cost four.
        
        \vspace{1mm}
    \setlength\tabcolsep{1.5pt} 
        \begin{tabular}{|c||c|c|c|c|c|}
        \hline
        & $\bfa^{(0)}$ & $\bfa^{(1)}$ & $\bfa^{(2)}$ & 
        $\bfp^{(0)}$ & $\bfp^{(1)}$ \\
        & $\spadesuit$  & $\heartsuit$ &  $\clubsuit$ &  
        & \\ \hline 
        $00$ & {\color{blue}$00_\spadesuit$} & {\color{blue}$00_\heartsuit$}  & {\color{red}$00_\clubsuit$} & 
        {\color{blue}$00_{ \spadesuit }\boxplus 00_{\heartsuit }\boxplus 00_{\clubsuit}$} & 
        {\color{blue}$00_{ \spadesuit }\boxplus 10_{\heartsuit }\boxplus 01_{\clubsuit}$} \\ 
        
        $01$ & $01_\spadesuit$ & $01_\heartsuit$  & {\color{red}$01_\clubsuit$} &
        $01_{ \spadesuit }\boxplus 01_{\heartsuit}\boxplus 01_{\clubsuit}$ &
        $01_{ \spadesuit }\boxplus 11_{\heartsuit}\boxplus 00_{\clubsuit}$ \\ 
        
        $10$ & {\color{blue}$10_\spadesuit$} & {\color{blue}$10_\heartsuit$}  & {\color{red}$10_\clubsuit$} & 
        {\color{blue}$10_{ \spadesuit }\boxplus 10_{\heartsuit }\boxplus 10_{\clubsuit}$} &
        {\color{blue}$10_{ \spadesuit }\boxplus 00_{\heartsuit }\boxplus 11_{\clubsuit}$} \\ 
        
        $11$ & $11_\spadesuit$ & $11_\heartsuit$  & {\color{red}$11_\clubsuit$} & 
        $11_{ \spadesuit }\boxplus 11_{\heartsuit }\boxplus 11_{\clubsuit}$ & 
        $11_{ \spadesuit }\boxplus 01_{\heartsuit }\boxplus 10_{\clubsuit}$ \\ \hline
        \end{tabular}
    \end{center}
    
        \vspace{2mm}
        \begin{center}
        \noindent(b) $(6,3,4)$-array code with skip cost zero.
        \vspace{1mm}
        
    {\centering \scriptsize
        \setlength\tabcolsep{1pt} 
        \begin{tabular}{|c||c|c|c|c|c|c|}
        \hline
        & $\bfa^{(0)}$ & $\bfa^{(1)}$ & $\bfa^{(2)}$ & 
          $\bfp^{(0)}$ & $\bfp^{(1)}$ & $\bfp^{(2)}$ \\
        & $\spadesuit$  & $\heartsuit$ &  $\clubsuit$ &  
      & &  \\ \hline 
        $00$ & {\color{blue}$00_\spadesuit$} & {\color{blue}$00_\heartsuit$}  & {\color{red}$00_\clubsuit$} & 
        {\color{blue}$00_{ \spadesuit }\boxplus 00_{\heartsuit }\boxplus 00_{\clubsuit}$} & 
        {\color{blue}$00_{ \spadesuit }\boxplus 01_{\heartsuit }\boxplus 10_{\clubsuit}$} &
        $00_{ \spadesuit }\boxplus 10_{\heartsuit }\boxplus 11_{\clubsuit}$  \\ 
        
        $01$ & {\color{blue}$01_\spadesuit$} & {\color{blue}$01_\heartsuit$}  & {\color{red}$01_\clubsuit$} &
        {\color{blue}$01_{ \spadesuit }\boxplus 01_{\heartsuit}\boxplus 01_{\clubsuit}$} &
        {\color{blue}$01_{ \spadesuit }\boxplus 00_{\heartsuit}\boxplus 11_{\clubsuit}$} &
        $01_{ \spadesuit }\boxplus 11_{\heartsuit }\boxplus 10_{\clubsuit}$  \\ 
        
        $10$ & $10_\spadesuit$ & $10_\heartsuit$  & {\color{red}$10_\clubsuit$} & 
        $10_{ \spadesuit }\boxplus 10_{\heartsuit }\boxplus 10_{\clubsuit}$ &
        $10_{ \spadesuit }\boxplus 11_{\heartsuit }\boxplus 00_{\clubsuit}$ &
        $10_{ \spadesuit }\boxplus 00_{\heartsuit }\boxplus 01_{\clubsuit}$  \\ 
        
        $11$ & $11_\spadesuit$ & $11_\heartsuit$  & {\color{red}$11_\clubsuit$} & 
        $11_{ \spadesuit }\boxplus 11_{\heartsuit }\boxplus 11_{\clubsuit}$ & 
        $11_{ \spadesuit }\boxplus 10_{\heartsuit }\boxplus 01_{\clubsuit}$ &
        $11_{ \spadesuit }\boxplus 01_{\heartsuit }\boxplus 00_{\clubsuit}$  \\  \hline
        \end{tabular}}
        \end{center}
        \caption{(a)~Example of a {$(5,3,4)$}-MDS array code constructed in~\cite{Tamo2013}. Suppose information node $\bfa^{(2)}$ (highlighted in {\color{red}red}) fails. We contact nodes $\bfa^{(0)}$, $\bfa^{(1)}$, $\bfp^{(0)}$ and $\bfp^{(1)}$ and read the contents in {\color{blue}blue}. Here, the skip cost is $4\times 1 = 4$. 
        (b)~Example of a {$(6,3,4)$}-MDS array code described in \cite{chee2024repairing } with $m=2$. Suppose information node $\bfa^{(2)}$ (highlighted in {\color{red}red}) fails. We contact nodes $\bfa^{(0)}$, $\bfa^{(1)}$, $\bfp^{(0)}$ and $\bfp^{(1)}$ and read the contents in {\color{blue}blue}. Here, the skip cost is zero.
        Note that we use $\bfx_{i}$ to represent the information symbol $a_\bfx^{(i)}$, while the `sum' $\bfx_i\boxplus \bfy_j\boxplus \bfz_k$ indicates that the corresponding codesymbol is a linear combination of $a_\bfx^{(i)}$, $a_\bfy^{(j)}$, and $a_\bfz^{(k)}$. 
        }\label{fig:zigzag}
    \end{figure}

\subsection{Related works}

Many MSR codes with optimal repair bandwidth and optimal access property have been proposed in recent years \cite{Tamo2013,wang2016explicit,ye2017explicit,li2018generic,liu2018explicit,li2024msr}. 
When the number of helper nodes $d=N-1$, the {\em zigzag codes} with parameters $(k+p,k,M=p^{k-1})$ and $(k+p,k,M=p^{k+1})$ proposed in~\cite{Tamo2013,wang2016explicit} achieve optimal access and operate over the smallest known finite fields, namely $\mathbb{F}_3$ and $\mathbb{F}_4$ for $p=2$ and $p=3$, respectively. However, the field size required to construct zigzag codes for $p > 4$ remains unknown. 
Later, Ye et al.~\cite{ye2017explicit} constructed a $(k+p,k,M=p^{k+p-1})$ MSR code over a finite field of size $|\bbF|\geq k+p-1$.  
The generic transformation in~\cite{li2018generic} was applied to an MDS code to obtain a $ (k+p,k,M=p^{\lceil\frac{k+p}{p}\rceil})$ MSR code over a finite field of size $q\geq k+p$.
Later, Liu et al.~\cite{liu2018explicit} constructed optimal-access MSR codes  over small fields  under the condition $(q-1)\nmid(p+1)$, e.g., $\mathbb{F}_3$ for even $p$ and $\mathbb{F}_q$ with $q \ge p + 1$ for odd $p$. 
Li et al.~ \cite{li2024msr} recently proposed MSR codes with the smallest known subpacketization $M= s^\frac{N}{s}$ for $k+1\leq d\leq N-1$, over a field of size $q \ge Ns + (s - 1)2^{s - 2}$, where $s = d - k + 1$.
These constructions are subject to a known lower bound on the subpacketization. Specifically,
~\cite{tamo2014access} showed that any optimal-access MSR code with helper nodes $d=N-1$ must satisfy  $M\geq p^{\frac{k-1}{p}}$,  which  was later generalized to $M\geq s^\frac{N}{s}$  for  $s=d-k+1$ in~\cite{balaji2022lower}. 
These  result indicates that for MSR codes with optimal-access, the subpacketization level $M$ must grow rapidly with the number of information nodes $k$. Therefore, for a fixed code length $N$, it is impossible for MSR codes with optimal access to achieve both subpacketization and field size to be subexponential in $N$ simultaneously. 

In addition to reducing subpacketization, minimizing the skip cost (i.e., the number of non-contiguous disk sections accessed during repair) is crucial for improving I/O efficiency during degraded reads. As noted in~\cite{wu2021}, data access time is influenced not only by the amount of data read but also by the physical continuity of the accessed symbols on disk%
\footnote{In fact, this issue was raised during a Huawei Industry Session during ISIT 2023. Dr Wu discussed this in his presentation ``Coding Challenges for Future Storage''.}.
Codes with zero skip cost are also referred to as \emph{degraded-read-friendly}. Several recent works have contributed to the construction of such codes. When $p=2$ and $M=2$, Wu et al.~\cite{wu2021} derived a lower bound on the \emph{access bandwidth} for degraded-read-friendly MDS array codes, 
and proposed an explicit construction that achieves this bound.
Later, Liu and Zhang~\cite{liu2024family} proposed a family of $(k+p, k, M=p^2)$ degraded-read-friendly MDS array codes.
Independently, Chee et al.~\cite{chee2024repairing} (prior work of a subset of the current authors%
\footnote{This paper was presented in part at the 2024
International Symposium on Information Theory~\cite{chee2024repairing}}
) introduced the skip cost metric and constructed both fractional repetition codes and zigzag codes with a coding rate of $1/2$ that achieve zero skip cost. 
Most recently, Yu et al.~\cite{yu2025zero} demonstrated that zero skip cost is always attainable for fractional repetition codes using covering designs, and presented several constructions that satisfy this property. Thus, while degraded-read-friendly codes have been demonstrated in special cases, no general framework exists for achieving low skip cost in zigzag codes with flexible parameters.

\subsection{Our Contributions}

\begin{itemize}
\item Motivated by the difficulty of achieving small subpacketization and small field size simultaneously in MSR codes with optimal access, we focus on zigzag array codes under more relaxed conditions.
In this work, we revisit the classical zigzag construction and analyze its MDS property. By assigning Cauchy-type coefficients to the generator matrix, we show that zigzag codes constructed over $G = \mathbb{F}_2^m$ achieve the MDS property over a finite field $\bbF$  of characteristic two with size satisfying $|\mathbb{F}| \geq N$. Notably, our result allows the number of information nodes $k$ and the subpacketization $M$ to be chosen independently, thereby enabling the construction of MDS array codes that simultaneously achieve {\em both small subpacketization and small field size} -- a feat impossible under the MSR requirement.

\item Moreover, we propose explicit constructions of zigzag codes with {\em bounded skip cost} and, consequently, {\em low repair-fragmentation ratio (RFR)}, while preserving both optimal access and optimal rebuilding ratio for single-node repair. 
To this end, we revisit the zigzag framework and introduce algebraic conditions that enable precise skip cost analysis. 
Our constructions include a family of codes with zero skip cost whose rates approach $2/3$, as well as families with bounded skip cost whose rates approach $3/4$ and $4/5$. 
Importantly, the rates of these latter families are comparable to those of array codes deployed in practice (see Table~\ref{tab:parameters}).

\end{itemize}
The rest of this paper is organized as follows. Section~\ref{sec:preli} reviews concepts related to MDS array codes and formalizes repair metrics, including the new repair-fragmentation ratio (RFR). 
Section~\ref{sec:zigzag} presents the zigzag code framework and restates the parity construction and recovery conditions using a group-structure notation. 
Section~\ref{sec:MDSproperty} shows that zigzag codes can achieve the MDS property over small fields when the number of information nodes $k$ and the subpacketization level $M$ are chosen independently. Section~\ref{sec:skipcost} proposes an explicit construction of zigzag codes with zero skip cost. Finally, Section~\ref{sec:conclusion} summarizes our contributions.

\section{Preliminaries}\label{sec:preli}

Throughout this paper, 
let $[i,j]$ denote the set $\{i, i+1, \ldots, j\}$ for integers $i < j$, and let $[i]$ denote the set $\{1, 2, \ldots, i\}$. Furthermore, $\bbF_q$ denotes the finite field with $q$ elements.
Furthermore, matrices and vectors are highlighted in boldface with the following notation. For any positive $m$, denote by $\bfI_m$ the identity matrix of size $m$. $\vzero$ and $\mathbf{1}$ denote the zero vector and all-ones vector, respectively. For $i\in [m]$, we use $\bfe_{i}$ to denote the vector whose $i$-th entry is one and other entries are $0$.

\subsection{MDS Array Code}

Let $\mathcal{C} \subseteq \mathbb{F}_q^{M \times N}$ be an $(N, k, M)$ array code. A codeword $\bfC\in\mathcal{C}$ takes the form $\bfC = \left[\bfb^{(1)}, \bfb^{(2)}, \ldots, \bfb^{(N)}\right]$, where each column $\bfb^{(j)} = \left(b_1^{(j)}, \ldots, b_M^{(j)}\right)$ belongs to $\bbF_q^M$ and corresponds to the data stored in node $j$ where $j \in [N]$.
The code $\mathcal{C}$ is defined via a generator matrix $\bfG \in \mathbb{F}_q^{kM \times NM}$, which is partitioned into $N$ column blocks as:
\begin{equation}\label{eq:gen_mat}
\bfG = 
\begin{bmatrix}
\bfG_1 & \bfG_2 & \cdots & \bfG_N
\end{bmatrix}
=
\begin{pmatrix}
\bfG_{1,1} & \bfG_{1,2} & \cdots & \bfG_{1,N} \\
\bfG_{2,1} & \bfG_{2,2} & \cdots & \bfG_{2,N} \\
\vdots     & \vdots     & \ddots & \vdots     \\
\bfG_{k,1} & \bfG_{k,2} & \cdots & \bfG_{k,N}
\end{pmatrix},
\end{equation}
where each block $\bfG_{i,j} \in \mathbb{F}_q^{M \times M}$ for $i \in [k]$ and $j \in [N]$. Each column block $\mathbf{G}_j$ serves as the encoding matrix for node $j$.

Let $\bfm \in \mathbb{F}_q^{kM}$ be the information vector formed by stacking $k$ blocks $\bfm^{(i)} \in \mathbb{F}_q^M$, i.e.,
\[
\bfm =
\begin{pmatrix}
\bfm^{(1)} \\
\bfm^{(2)} \\
\vdots \\
\bfm^{(k)}
\end{pmatrix},
\] where each $\bfm^{(i)} \in \mathbb{F}_q^M$ consists of the $M$ symbols stored in the $i$-th systematic node. The codeword $\mathbf{C} = [\bfb^{(1)}, \ldots, \bfb^{(N)}]$ is therefore generated as
$\bfb^{(i)} = \bfG_i^\top \bfm, \quad \forall\, i\in [N]$.

Throughout this work, we focus on {\em systematic} array codes $\mathcal{C}$, where for all $i \in [k]$, the $i$-th node directly stores the message block $\bfm^{(i)}$, i.e., $\bfb^{(i)} = \bfm^{(i)}$.

Moreover, the code $\mathcal{C}$ is said to have \emph{MDS property} if the original file can be recovered from any $k$ out of the $N$ nodes. Equivalently, we require that certain $kM\times kM$ submatrices of the generator matrix have full rank.

\begin{lemma}[MDS Property]\label{lem:MDS}
An $(N, k, M)$ array code defined by generator matrix of the form~\eqref{eq:gen_mat} has the MDS property if and only if for any index set $\mathcal{S} \subseteq [N]$ with $|\mathcal{S}| = k$, the corresponding block submatrix  $\bfG_\mathcal{S}$ of the generator matrix,
\begin{equation}
\bfG_\mathcal{S}^\top = 
\begin{bmatrix}
\bfG_{i_1}^\top \\
\vdots \\
\bfG_{i_k}^\top
\end{bmatrix}\in \mathbb{F}_q^{kM \times kM},
\end{equation}
formed by concatenating the block-columns $\bfG_i$ for $i \in \mathcal{S}$, must be full-rank, i.e.,
$\operatorname{rank}(\bfG^\top_\mathcal{S}) = kM$.
\end{lemma}

\subsection{Repair Bandwidth and Rebuilding Ratio}

For an $(N,k,M)$ array code, a key metric is the \emph{repair bandwidth}, that is, the total amount of information downloaded from the helper nodes during the repair of a failed node. 
The fundamental cut-set bound of Dimakis \textit{et al.}~\cite{dimakis2010network} shows that the repair bandwidth $\gamma(d)$ obeys the following inequality
\begin{equation}\label{eq:cutset}
    \gamma(d) \;\geq\; \gamma^*(d) \triangleq \frac{d}{d-k+1}\,M,
\end{equation}
where $d \in [k,N-1]$ denotes the number of helper nodes contacted. 

\begin{definition}[Optimal Repair and Access Properties]
An $(N,k,M)$ MDS array code is said to have  
\begin{itemize}
    \item \emph{optimal repair property} if the equality 
    $\gamma(d) = \gamma^*(d)$ holds in~\eqref{eq:cutset}, that is, each helper contributes exactly $M/(d-k+1)$ symbols, and
    \item \emph{optimal access property} if these symbols are accessed directly without additional computation.
\end{itemize}
\end{definition}

A convenient way to benchmark repair bandwidth efficiency is via the \emph{rebuilding ratio}.  

\begin{definition}[Rebuilding Ratio]
Let node $i$ be repaired by contacting a set of helpers $\cH_i$ with $|\cH_i|=d$. 
For each helper $j \in \cH_i$, let $\Read_{i,j} \subseteq [M]$ denote the indices of symbols read from node $j$. 
The \emph{rebuilding ratio} of the code is
\[
\rho \;=\; \max_{i \in [k]} \; \max_{j \in \cH_i} \;\frac{|\Read_{i,j}|}{M}.
\]
\end{definition}

Thus, the rebuilding ratio quantifies the worst-case load on a single helper. It follows from \eqref{eq:cutset} that codes with optimal access property have rebuilding ratio at least $1/(d-k+1)$.

\subsection{Skip Cost and Repair-Fragmentation Ratio}

Beyond the amount of data read, an equally important factor in repair performance is the \emph{contiguity} of these reads within each helper node. 
When symbols are scattered, the repair process may incur higher I/O latency even if the total number of symbols accessed is small. 
We quantify this phenomenon using the notion of \emph{skip cost}.

\begin{definition}[Skip Cost]\label{def:skipcost}
Let $\bfb=(b_1,b_2,\ldots,b_M)$ denote the contents of a helper node, and suppose that a repair scheme reads symbols 
$\Read=\{b_{i_1},\ldots,b_{i_t}\}$ with $i_1 < \cdots < i_t$. 
The \emph{skip cost} of $\Read$ is 
\[
\cost(\Read)\triangleq i_t - i_1 - (t-1).
\]
If the indices in $\Read$ form a consecutive block, then $\cost(\Read)=0$, which is optimal.
\end{definition}

\begin{definition}[Skip Cost of a Code]
An array code has a repair scheme with \emph{skip cost $\sigma$} if, for any failed node $\bfb$, there exists a helper set $\cH$ and corresponding read sets $\{\Read_h : h \in \cH\}$ such that
\begin{enumerate}[(i)]
    \item $\bfb$ can be reconstructed from $\bigcup_{h\in\cH} \Read_h$, and
    \item $\sum_{h\in\cH} \cost(\Read_h) \le \sigma$.  
\end{enumerate}
In particular, if $\sigma=0$, we say the code has \emph{zero skip cost}.
\end{definition}

While skip cost captures the fragmentation within each helper, it does not account for the total number of symbols read. 
To balance these two factors, we introduce the following normalized measure.

\begin{definition}[Repair-Fragmentation Ratio (RFR)]
Let $\cC$ be an $(N,k,M)$ array code. For the repair of node $i \in [N]$ from a helper set $\cH_i$, let $\Read_{i,j}$ denote the indices read from helper $j$, and let $\cost(\Read_{i,j})$ be the corresponding skip cost. 
The \emph{repair-fragmentation ratio} of the code is
\[
\mathsf{RFR} \;=\; 
\max_{i \in [N]} \;
\frac{\displaystyle \sum_{j \in \cH_i} \Big(|\Read_{i,j}| \;+\; \cost\left(\Read_{i,j}\right)\Big)}{NM}.
\]
\end{definition}

To evaluate the efficiency of repair schemes, we now have three complementary metrics. 
The \emph{rebuilding ratio} measures the fraction of contents read from the busiest helper, capturing load balance. 
The \emph{skip cost} measures the fragmentation of reads within each helper, with zero skip cost corresponding to contiguous access. 
Finally, the \emph{repair-fragmentation ratio (RFR)} combines both dimensions by normalizing the total reads and skips against the overall packet size. 
Together, these metrics provide a comprehensive view of the repair performance of an array code.

\subsection{Zigzag Codes and MSR Codes}
\label{sec:zigzag-msr}
Zigzag codes, introduced by Tamo, Wang, and Bruck~\cite{Tamo2013}, are a well-known family of $(N,k,M)$ MDS array codes that achieve the optimal access property and the optimal rebuilding ratio for single systematic node failures. 
Specifically, with $p=N-k$ parity nodes and subpacketization $M=p^m=(N-k)^m$, zigzag codes achieve an optimal rebuilding ratio of $1/p$. 
For the special cases $p=2$ and $p=3$, the required finite field size can be as small as $3$ and $4$, respectively. 
For $p=2$, when $k$ is chosen independently of $M$, it suffices for the field size to exceed $M(k-1)+1$ to ensure the MDS property~\cite{Tamo2013}.

A closely related class of codes is the Minimum Storage Regenerating (MSR) codes of Dimakis \textit{et al.}~\cite{dimakis2010network}. 
For the case $d = N-1$, explicit MSR codes with the optimal access property have been constructed in~\cite{ye2017explicitnearly,sasidharan2016explicit,li2018generic} with subpacketization $\ell = p^{\lceil n/p \rceil}$ over fields of linear size. 
The fundamental limitation, however, is that small field size and small subpacketization cannot be achieved simultaneously for MSR codes, as shown in~\cite{balaji2022lower}. 
Later constructions such as~\cite{li2024msr} aim to match this lower bound as closely as possible, achieving optimal access and rebuilding ratio over small fields but at the expense of exponentially large subpacketization. 

In this work, we take a different approach: rather than insisting on the MSR requirement, we relax it and ask whether one can still achieve an optimal rebuilding ratio under this broader setting. 
By focusing on zigzag codes with $M=2^m$ over fields of characteristic two, we show that the MDS property holds for arbitrary $p$, $k$, and $M$ whenever $|\bbF|\ge N$. 
This significantly extends the range of achievable parameters beyond what was known in~\cite{Tamo2013}. 
Moreover, we introduce a new algebraic framework based on explicit orderings and subgroup structures—tools not considered in the original zigzag construction—that allow us to design repair-by-transfer schemes with bounded skip cost and low repair-fragmentation ratio, while still retaining optimal access and optimal rebuilding ratio for single-node repair. 
Thus, our work demonstrates that zigzag codes can be both more flexible and more efficient than previously thought.

\section{The Zigzag Code Framework}
\label{sec:zigzag}

In this section, we revisit the zigzag framework introduced by Tamo, Wang, and Bruck~\cite{Tamo2013}. 
We describe in detail how parity symbols are constructed and provide algebraic conditions under which a failed node can be recovered.
While conditions similar to Lemma~\ref{lem:recovery} appeared in Tamo–Wang–Bruck, we restate and slightly adapt them here to keep the exposition self-contained and aligned with our notation.

Let $\cC$ be an $(N, k, M)$ systematic MDS array code over the finite field $\bbF_q$, where $k$ denotes the number of systematic nodes and $p = N - k$ denotes the number of parity nodes. 
To facilitate an algebraic description of the code construction, particularly for defining parity symbols and analyzing repair properties, we index the rows of the array using elements from a finite additive group, referred to as the \emph{row-index group}. Specifically, let $\cG$ be the additive group of order $M$, and label each row of the array by an element $g \in \cG$. Each systematic symbol is denoted $a^{(j)}_g$, where $g \in \cG$ and $j \in [k] = \{1,2,\ldots,k\}$. 
For each parity node $i \in [p]$, we associate a permutation label
\begin{equation}\label{eq:perm_label}
\bfS_i = \left( u_1^{(i)}, u_2^{(i)}, \dotsc, u_k^{(i)} \right), \quad \text{with } u_j^{(i)} \in \cG.
\end{equation}
Then, in row $g \in  \cG $ of the $i$-th parity column, we store the symbol 
\begin{equation}\label{checknode:def}
\sum_{j \in [k]} \gamma_{g,j}^{(i)} \, a^{(j)}_{g + u_j^{(i)}},
\end{equation}
where $\gamma_{g,j}^{(i)} \in \bbF$ are coefficients to be specified. 

\begin{remark}
For any fixed $j \in [k]$ and $g \in \cG$, notice that the map $g \mapsto g + u_j^{(i)}$ defines a permutation on the row indices. This means that each parity symbol is formed by combining one symbol from each information column, taken at row positions that are permuted according to the label $\bfS_i$.
\end{remark}

    \begin{remark}
        In the rest of the paper, we omit the coefficients $\gamma_{g,j}^{(i)}$ and represent each parity symbol more compactly as
\begin{equation}\label{eq:parity_symbol}
        \bigboxplus_{j \in [k]} a^{(j)}_{g + u_j^{(i)}},
\end{equation}
Indeed, prior work~\cite{Tamo2013,chee2024repairing} showed that when the field size is sufficiently large, it is possible to choose coefficients that guarantee the MDS property. In this work, we strengthen that result. In Section~\ref{sec:MDSproperty}, we show that under mild assumptions, such coefficients always exist whenever $|\bbF| \ge N$. Moreover, we provide an explicit expression for these coefficients.
\end{remark}

Next, we examine how the array code enables efficient node repair. In particular, we identify conditions under which the failure of a single systematic node can be resolved by accessing only two parity nodes. The following lemma provides a sufficient condition for such recovery based on the group-theoretic structure of the permutation labels.

\begin{lemma}[Recovery Lemma]\label{lem:recovery}
Consider an $(N, k, M)$ systematic array code with rows indexed by elements of an additive group $\cG$. Suppose that the $j$-th systematic node, for some $j \in [k]$, has failed. Let the two parity nodes be accessed:
\begin{itemize}
    \item $p_s$, with permutation label $\bfS_s = (u_1, u_2, \ldots, u_k)$,
    \item $p_t$, with permutation label $\bfS_t = (v_1, v_2, \ldots, v_k)$,
\end{itemize}
where each $u_i, v_i \in \cG$. Suppose there exists a subgroup $H \subseteq \cG$ of index two, i.e., $|H| = |\cG|/2$, such that one of the following two cases holds: 
\begin{description}
    \item[Case 1:]~
    \begin{enumerate}[(R1)]
        \item The pair $\{u_j, v_j\}$ forms a complete set of coset representatives of $H$ in $\cG$;
        \item For every $i \ne j$, the pair $\{u_i, v_i\}$ lies entirely within a single coset of $H$.
    \end{enumerate}
    
    \item[Case 2:]~
    \begin{enumerate}[(R1$'$)]
        \item The pair $\{u_j, v_j\}$ lies entirely within a single coset of $H$;
        \item For every $i \ne j$, the pair $\{u_i, v_i\}$ forms a complete set of coset representatives of $H$ in $\cG$.
    \end{enumerate}
\end{description}
Then, the $j$-th systematic node can be recovered by downloading $|H|$ symbols from each of the $k - 1$ surviving systematic nodes and from the two parity nodes $p_s$ and $p_t$.
\end{lemma}

\begin{proof}
Since $H$ is a subgroup of index two, we can partition $\cG$ into two cosets: $H$ and $\cG \setminus H$. 
Then, by definition, the parity-check symbols in row $g \in \mathcal{G}$ of nodes $p_s$ and $p_t$ are given by:
\[
\bigboxplus_{z \in [k]} a^{(z)}_{g + u_z}, \quad \text{and} \quad \bigboxplus_{z \in [k]} a^{(z)}_{g + v_z},
\]
respectively.

We first consider Case 1. To recover the information at node $j$ using the two parity nodes $p_s$ and $p_t$ and the remaining $(k - 1)$ systematic nodes, we proceed as follows:

\begin{itemize}
    \item From each surviving systematic node $i \ne j$, we download the symbols $a^{(i)}_{g + u_i}$ for all $g\in H$. 
    \item From parity node $p_s$, we download the symbols
    $\bigboxplus_{z \in [k]} a^{(z)}_{g + u_z}$ for all $g\in H$.
    
    \item From parity node $p_t$, we download the symbols
    $\bigboxplus_{z \in [k]} a^{(z)}_{g + v_z}$ for all $g \in H$.
\end{itemize}

Next, it remains to recover all information symbols of node $j$ from these downloaded symbols. Now, from (R1), since $\{u_j, v_j\}$ forms a complete set of coset representatives of $H$, every symbol $a_{g}^{(j)}$ belongs to one of the two cosets. Hence, the symbol is either of the form $a_{g + u_j}^{(j)}$ or $a_{g + v_j}^{(j)}$ for some $g \in H$.

\begin{itemize}
    \item  Recovery of $a^{(j)}_{g + u_j}$ for  $g \in H$: we observe that the parity symbol $\boxplus_{z \in [k]} a^{(z)}_{g + u_z}$ from parity node $p_s$ includes the unknown $a^{(j)}_{g + u_j}$, and all other terms $a^{(i)}_{g + u_i}$ for $i \ne j$ can be obtained from the remaining information nodes.

    \item  Recovery of $a^{(j)}_{g + v_j}$ for  $g \in H$: similar to before, the parity symbol $\boxplus_{z \in [k]} a^{(z)}_{g + v_z}$ from parity node $p_t$ includes the unknown $a^{(j)}_{g + v_j}$. 
    As before, the remaining terms are known, so we recover the desired symbol.
    \end{itemize}
    
For Case 2, according to (R1$'$), ${u_j, v_j}$ lies entirely within a single coset of $H$. Without loss of generality, suppose $u_j, v_j \in H$. The recovery proceeds as follows:
\begin{itemize}
    \item From each surviving systematic node $i \ne j$, we download the symbols $a^{(i)}_{g + u_i}$ for all $g\in H$. 
    \item From parity node $p_s$, we download the symbols
    $\bigboxplus_{z \in [k]} a^{(z)}_{g + u_z}$ for all $g\in  H$.
    \item From parity node $p_t$, we download the symbols
    $\bigboxplus_{z \in [k]} a^{(z)}_{g + v_z}$ for all $g\in \cG\setminus H$.
\end{itemize}

As $g$ runs over $H$ and $\cG\setminus H$, the sets 
$\{a_{g+u_j}^{(j)} : g\in H\}$ and 
$\{a_{g+v_j}^{(j)} : g\in \cG\setminus H\}$ 
together cover all symbols of node~$j$. 
Moreover, for any surviving node $i \ne j$, condition (R2$'$) ensures that 
$\{u_i,v_i\}$ forms a complete set of coset representatives of $H$ in $\cG$, 
so $v_i-u_i \notin H$. 
Hence, for any $g \in \cG \setminus H$ we can set 
$h = g+(v_i-u_i)\in H$, which gives
$a^{(i)}_{g+v_i} = a^{(i)}_{h+u_i}$,
and the right-hand side belongs to the downloaded set 
$\{ a^{(i)}_{g+u_i} : g\in H\}$. 
Therefore, together with the directly downloaded symbols from the surviving systematic nodes, every $a^{(j)}_{g}$ for $g\in\cG$ can be recovered in Case~2.
\end{proof}

\begin{remark}
Condition (R1) ensures that the information symbols at a failed node $j$
 can be fully recovered.  Specifically, if $\{u_j, v_j\}$ is a complete set of coset representatives for $H$ in $\mathcal{G}$, then the collections $\{a^{(j)}_{g+u_j} \mid g \in H\}$ and $\{a^{(j)}_{g+v_j} \mid g \in H\}$ together form the complete set of symbols at node $j$. Because $H$ has an index of two in $\mathcal{G}$, the same property holds when indexing over the other coset $\mathcal{G} \setminus H$. 
 This is because if we take $u_j \in H$ and $v_j \in \mathcal{G} \setminus H$, the index sets become $\{u_j + g \mid g \in \mathcal{G} \setminus H\} = \mathcal{G} \setminus H$ and $\{v_j + g \mid g \in \mathcal{G} \setminus H\} = H$,  which together cover all symbols.
Moreover, condition (R2) ensures that all symbols indexed by $g \in H$ downloaded from the two parity nodes are composed of the symbols indexed by $g \in H$ that are downloaded from each surviving systematic node. The same property holds for symbols  indexed by $g \in \cG \setminus H$ downloaded from the systematic nodes $i\neq j$ and the two parity nodes $p_s,p_t$. Therefore, in the above proof, downloading the symbols
$ a^{(i)}_{g + u_i},\quad\bigboxplus_{z \in [k]} a^{(z)}_{g + u_z},\quad
\bigboxplus_{z \in [k]} a^{(z)}_{g + v_z}
$ from nodes $i\neq j$, $p_s$ and $p_t$ for all $g \in \cG \setminus H$ also enables the recovery of failed node $j$. Similarly, the same arguments apply when (R1$'$) and (R2$'$) are satisfied.

\end{remark}

We give a simple example below to show how the lemma~\ref{lem:recovery} works.

    \begin{example}\label{ex:1}
  Let $m=2$ and consider the additive group $  \mathcal{G} = \bbZ_2^m$, so that $  |\mathcal{G}| = 4$. Let \( H_1 = \{00, 01\} \) and $ H_2 = \{00, 11\} $ be two subgroups of \( \mathcal{G} \), each of index two.   Consider an  array code  $\cC$ with $N = 7$ nodes, consisting of $k = 4$ systematic nodes and $p = 3$ parity nodes.  The rows of the array are indexed by elements $\bfg \in \mathcal{G}$.
Each systematic symbol is denoted by \( \bfa^{(j)}_{\bfg} \), where \( j \in [4] \) and \( \bfg \in \cG \). For notational simplicity, we write \( a^{(j)}_{\bfg} \) as \( \bfg_j \). For example, \( 00_1 \) denotes the symbol \( a^{(1)}_{00} \), stored in node 1 and row \( 00 \). The  permutation label associated with the parity nodes $\bfp^{(i)}$ for $i\in[3]$, is given as follows:

\begin{equation*}
\bfS_1 = (\vzero, \vzero, \vzero, \vzero), \bfS_2 = (\vone, \bfe_2, \vzero,\vzero),\bfS_3 = ( \vzero, \vzero, \vone,\bfe_2).
\end{equation*}

\noindent where $\vzero=(0,0)$, $\vone=(1,1)$, and $\bfe_2=(0,1)$ are elements in the additive group $\cG$. For parity symbols, we use $\bfg_i \boxplus \bfx_j \boxplus \bfy_k$ to represent a linear combination of $a_\bfg^{(i)}$, $a_\bfx^{(j)}$, and $a_\bfy^{(k)}$. 
For example, with the permutation label $\bfS_1 = (\vzero, \vzero, \vzero, \vzero)$, the first row stored in $\bfp^{(1)}$ is  $00_1 \boxplus 00_2 \boxplus 00_3 \boxplus 00_4$.
The placement of code symbols is  shown in Fig~\ref{fig:zigzagk4}.
            \begin{figure*}[ht]
                \centering
                \setlength\tabcolsep{1pt} 
                \begin{tabular}{|c|c|c|c|c|c|c|c|c|c|c|}
                    \hline
              \multirow{2}{*}{$\bfa^{(1)}$} & \multirow{2}{*}{${\color{red}\bfa^{(2)}}$} & \multirow{2}{*}{$\bfa^{(3)}$}&\multirow{2}{*}{$\bfa^{(4)}$}   & 
                      $\bfp^{(1)}$ & $\bfp^{(2)}$ & $\bfp^{(3)}$ \\
                  & & & &  $\bfS_1=(\vzero,\vzero,\vzero,\vzero)$ & $\bfS_2=(\vone,\bfe_2,\vzero,\vzero)$ & $\bfS_3=(\vzero,\vzero,\vone,\bfe_2)$\\ \hline 
                    $00_1$ & ${\color{red}00_2}$ & $00_3$ & $00_4$ & 
                $00_1 \boxplus 00_2 \boxplus 00_3 \boxplus 00_4$ &$11_1\boxplus01_2\boxplus 00_3\boxplus00_4$ & 
                $ 00_1 \boxplus 00_2 \boxplus 11_3\boxplus 01_4 $ \\
                
                {\color{blue}$01_1$}&      {\color{red}$01_2$} & {\color{blue}$01_3$} &      {\color{blue}$01_4$} & 
              {\color{blue}  $01_1 \boxplus 01_2 \boxplus 01_3 \boxplus 01_4$} &{\color{blue}$10_1 \boxplus 00_2\boxplus 01_3\boxplus 01_4$ } &
              $  01_1 \boxplus 01_2 \boxplus 10_3 \boxplus 00_4 $  \\
                
              {\color{blue}$10_1$} &      {\color{red}$10_2$} & {\color{blue}$10_3$} &      {\color{blue}$10_4$} & 
                    {\color{blue}$10_1 \boxplus 10_2 \boxplus 10_3 \boxplus 10_4$} 
                    &  {\color{blue}$ 01_1 \boxplus 11_2 \boxplus 10_3\boxplus 10_4$}
                    & $ 10_1 \boxplus 10_2 \boxplus 01_3\boxplus11_4$  \\
                
                    $11_1$ & ${\color{red}11_2}$ & $11_3$ & $11_4$ & 
                $11_1 \boxplus 11_2 \boxplus 11_3 \boxplus 11_4$ 
                & $ 00_1 \boxplus 10_2 \boxplus 11_3 \boxplus 11_4$ 
                & $ 11_1 \boxplus 11_2 \boxplus 00_3\boxplus 10_4 $   \\
                    \hline
                    \end{tabular}

                \caption{Example of a $(7, 4, 4)$ array code. }

                \label{fig:zigzagk4}
            \end{figure*}
    Assume that the information node $\bfa^{(2)}$ in Fig.~\ref{fig:zigzagk4} has failed. In our scheme, the helper information nodes are $\bfa^{(j)}$ for $j\in\{1,3,4\}$ while the helper parity nodes are $\bfp^{(1)}$ and $\bfp^{(2)}$. Specifically, the reads from the parity nodes contain the following information symbols. 
    
        \begin{itemize}
        \item From $\boldsymbol{p}^{(1)}$,  we obtain 
    \begin{align} \nonumber
    \{01_1,{\color{red}01_2},01_3,01_4\} \cup  \{10_{1},{\color{red}10_2},10_3,10_4\}. 
            \end{align}
    \item From $\boldsymbol p^{(2)}$, we obtain
     \begin{align} \nonumber
       \{10_1, {\color{red}00_2}, 01_3,01_4\}\cup
       \{ 01_1,{\color{red}11_2},10_3,10_4\}.
            \end{align}
    \end{itemize}
    
By combining the symbols $\{\bfg_i : \bfg \in \mathcal{G} \setminus H_2\}$ downloaded from $\bfa^{(i)}$ for $i \in [4] \setminus \{2\}$ together with the information obtained from the parity nodes $\bfp^{(1)}$ and $\bfp^{(2)}$, we can recover all symbols $\{\bfg_2 : \bfg \in \mathcal{G}\}$ of node $\bfa^{(2)}$.  Similarly, if the node $\bfa^{(1)}$ is erased, then from Fig.~\ref{fig:zigzagk4}, the helper information nodes are $\bfa^{(i)}$ for $i\in [4]\setminus \{1\} $, while the helper parity nodes are also $\bfp^{(1)}$ and $\bfp^{(2)}$.
      \begin{itemize}
        \item From $\boldsymbol{p}^{(1)}$, we obtain
    \begin{align} \nonumber
    \{{\color{red}00_1},00_2,00_3,00_4\} \cup  \{{\color{red}01_{1}},01_2,01_3,01_4\} .
            \end{align}
    \item From $\boldsymbol p^{(2)}$, we obtain
     \begin{align} \nonumber
 \{ {\color{red}11_1},01_2,00_3,00_4\} \cup
       \{{\color{red}10_1},00_2, 01_3,01_4\}.
            \end{align}
    \end{itemize}
 \noindent By downloading the information symbols
    $\{\boldsymbol{g}_{i}: \boldsymbol{g}\in H_1\}$ from $\bfa^{(i)}$ for $i\in [4]\setminus \{1\} $, it is clear that we can recover all values in $\{\boldsymbol{g}_{1}: \boldsymbol{g}\in\cG\}$ of node $\bfa^{(1)}$. 

    Furthermore, to repair the node $\bfa^{(3)}$, we download the symbols $\{\bfg_i : \bfg \in     H_1\}$ from each remaining information node $i \in [4] \setminus \{3\}$ and the parity nodes $\bfp^{(1)}$, $\bfp^{(3)}$, respectively. For the recovery of node $\bfa^{(4)}$, we download the symbols $\{\bfg_i : \bfg \in  \cG\setminus H_2
\}$ from each remaining information node $i \in [4] \setminus \{4\}$, and also download the same rows from the parity nodes $\bfp^{(1)}$ and $\bfp^{(3)}$.
    \end{example}

\begin{remark}
 In the work of TWB, the group $\cG$ was chosen as $\mathbb{Z}_p^m$, namely, the set of $m$-dimensional vectors over the ring of integers modulo $p$, where $p$ is the number of parity nodes, and $p^m$ is subpacketization level.  It is worth noting that our discussion in this work is restricted to subgroups of index two. The results in TWB also consider subgroups of arbitrary index in $\mathcal{G}$ (see, e.g., Lemma~14 in \cite{Tamo2013}).
\end{remark}


\section{Zigzag Codes with Row-Index Group 
$\cG=\bbZ_2^m$ admit MDS Coefficients Over Linear-Sized Fields of Characteristic Two}\label{sec:MDSproperty}


In this section, we show that zigzag codes admit the MDS property under mild conditions: 
(i) the row-index group is $\cG=\bbZ_2^m$, and 
(ii) the field $\bbF$ has characteristic two and size at least the array length $N$ (that is, $|\bbF|\ge N$).
Earlier small-field guarantees required the number of parity nodes to be restricted (namely, $p\in\{2,3\}$), in which case the row-index group was taken to be $\bbZ_p^m$\cite{Tamo2013}, coupling the subpacketization via $M=p^m$ to the redundancy.
In contrast, our approach works with $\cG=\bbZ_2^m$ for arbitrary $p=N-k$, so $k$ and $M$ can be chosen independently, and we obtain MDS array codes with both small field size and small subpacketization.

Recall that any $(N,k,M)$ array code $\cC$ has a (block) generator matrix $\bfG \in \bbF^{kM \times NM}$ given by~\eqref{eq:gen_mat}.
For a zigzag code, the $k\times N$ submatrices (each with dimension $M\times M$) satisfy the following properties:
$\bfG_{ii}=\bfI_M$ for $i\in[k]$, 
$\bfG_{ij}=\mathbf{0}$ for $i\ne j$ with $i,j\in[k]$, 
and for each parity column $j>k$ one has $\bfG_{ij}=\gamma_{ij}\,\bfA^\top_{ij}$, 
where $\bfA^\top_{ij}$ is a permutation matrix determined by the label $\bfS_{j-k}$ (see, for example,~\cite{Tamo2013}).


In the next theorem, we assign Cauchy-type coefficients to the permutation matrices in $\bfG$, yielding an $(N,k,2^m)$ MDS array code.

\begin{theorem}\label{thm:main}
For $ m \geq 1 $, let $\cG$ be the additive group $\bbZ_2^m $ and set $M = |\mathcal{G}| = 2^m $.
Let $p=N-k$, and let $\cC$ be an $(N,k,M)$ array zigzag code with $p$ parity nodes. 
Suppose the block generator matrix $\bfG^\top$ has the form
\begin{equation}\label{gen:G}
\bfG^\top = 
\begin{bmatrix}
    \bfI_M & \bf0 & \cdots & \bf0 \\
    \bf0 & \bfI_M & \cdots & \bf0 \\
    \vdots & \vdots & \ddots & \vdots \\
    \bf0 & \bf0 & \cdots & \bfI_M \\
    \gamma_{11}\,\bfA_{11} & \gamma_{12}\,\bfA_{12} & \cdots & \gamma_{1k}\,\bfA_{1k} \\
    \gamma_{21}\,\bfA_{21} & \gamma_{22}\,\bfA_{22} & \cdots & \gamma_{2k}\,\bfA_{2k} \\
    \vdots & \vdots & \ddots & \vdots \\
    \gamma_{p1}\,\bfA_{p1} & \gamma_{p2}\,\bfA_{p2} & \cdots & \gamma_{pk}\,\bfA_{pk}
\end{bmatrix},
\end{equation}
where each $\bfA_{ij}$ is a permutation matrix indexed by $\cG$.

Let $\alpha_1,\ldots,\alpha_p$ and $\beta_1,\ldots,\beta_k$ be $p+k=N$ distinct elements of a finite field $\bbF$ of characteristic two. If the Cauchy coefficients is given by
\begin{equation}\label{eq:cauchy-coefficients}
\gamma_{ij} = (\alpha_i - \beta_j)^{-1}, \qquad 1 \le i \le p,\ 1 \le j \le k,
\end{equation}
then $\cC$ is an $(N,k,M=2^m)$ MDS array code.
\end{theorem}

To prove this theorem, we first state a lemma that evaluates the determinant of a structured block matrix whose blocks are permutation matrices indexed by the row-index group $\cG=\bbZ_2^m$.

\begin{lemma}\label{lem:cauchy-det}
    Let $M$ and $N$ be positive integers, and let $\bbF$ be a finite field of characteristic two.
    Let $\alpha_1,\ldots,\alpha_N$ and $\beta_1,\ldots,\beta_N$ be $2N$ distinct elements of $\bbF$, and define $\gamma_{ij}=(\alpha_i-\beta_j)^{-1}$ for $1\le i,j\le N$.
    Let $\bfA$ be the $MN\times MN$ block matrix
    \[
    \bfA =
    \begin{bmatrix}
      \gamma_{11}\,\bfA_{11} & \gamma_{12}\,\bfA_{12} & \cdots & \gamma_{1N}\,\bfA_{1N} \\
      \gamma_{21}\,\bfA_{21} & \gamma_{22}\,\bfA_{22} & \cdots & \gamma_{2N}\,\bfA_{2N} \\
      \vdots & \vdots & \ddots & \vdots \\
      \gamma_{N1}\,\bfA_{N1} & \gamma_{N2}\,\bfA_{N2} & \cdots & \gamma_{NN}\,\bfA_{NN}
    \end{bmatrix},
    \]
    where each $\bfA_{ij}$ is an $M\times M$ permutation matrix indexed by $\cG=\bbZ_2^m$. If
    \[
    \Delta \;=\;
      \frac{\displaystyle\prod_{1\le j<i\le N}(\alpha_i-\alpha_j)(\beta_j-\beta_i)}
           {\displaystyle\prod_{i=1}^{N}\prod_{j=1}^{N}(\alpha_i-\beta_j)},
    \]
    then $\det(\bfA)=\Delta^{M}$.
\end{lemma}

When $M=1$ (equivalently, when $\cG$ is trivial), this reduces to the classical Cauchy determinant formula in~\cite{cauchy1841exercices}. 
The proof of Lemma~\ref{lem:cauchy-det} is technical, so we defer it to Appendix~A and proceed to apply Lemma~\ref{lem:cauchy-det} to prove Theorem~\ref{thm:main}.

\begin{proof}[Proof of Theorem~\ref{thm:main}]
To prove that the array code $\cC$ is MDS, it suffices to show that the original data can be reconstructed from any $k$ nodes; equivalently, every $k\times k$ block submatrix of $\bfG^\top$ must be full rank. Any such selection corresponds to a $k\times k$ block submatrix of $\bfG$, which may include both systematic and parity nodes. If all selected nodes are systematic, then the corresponding submatrix is block diagonal with identity blocks and hence invertible.

Otherwise, suppose the selection includes $t$ systematic nodes and $k-t$ parity-check nodes. Let $\{i_1,\ldots,i_t\}\subseteq [k]$ denote the indices of the selected systematic nodes, and let $\{j_1,\ldots,j_{\,k-t}\}\subseteq [p]$ denote the indices of the selected parity nodes. Define the set
$S \triangleq [k]\setminus\{i_1,\ldots,i_t\}=\{s_1,\ldots,s_{\,k-t}\}$.
After column and row permutations, 
the selected $k\times k$ block submatrix can be reduced to the form
\begin{equation*}
\begin{bmatrix}
    \bfI_{tM\times tM} & \bf0_{tM\times (k-t)M} \\
    \bf0_{(k-t)M\times tM} & \bfA_{(k-t)M\times (k-t)M}
\end{bmatrix}
=
\begin{bmatrix}
    \bfI_M & \bf0   & \cdots & \bf0   & \bf0   & \bf0   & \cdots & \bf0 \\
    \bf0   & \bfI_M & \cdots & \bf0   & \bf0   & \bf0   & \cdots & \bf0 \\
    \vdots & \vdots & \ddots & \vdots & \vdots & \vdots & \ddots & \vdots \\
    \bf0   & \bf0   & \cdots & \bfI_M & \bf0   & \bf0   & \cdots & \bf0 \\
    \bf0   & \bf0   & \cdots & \bf0   &
      \gamma_{j_1 s_1}\bfA_{j_1 s_1} & \gamma_{j_1 s_2}\bfA_{j_1 s_2} & \cdots & \gamma_{j_1 s_{k-t}}\bfA_{j_1 s_{k-t}} \\
    \bf0   & \bf0   & \cdots & \bf0   &
      \gamma_{j_2 s_1}\bfA_{j_2 s_1} & \gamma_{j_2 s_2}\bfA_{j_2 s_2} & \cdots & \gamma_{j_2 s_{k-t}}\bfA_{j_2 s_{k-t}} \\
    \vdots & \vdots & \ddots & \vdots &
      \vdots & \vdots & \ddots & \vdots \\
    \bf0   & \bf0   & \cdots & \bf0   &
      \gamma_{j_{k-t} s_1}\bfA_{j_{k-t} s_1} &
      \gamma_{j_{k-t} s_2}\bfA_{j_{k-t} s_2} & \cdots &
      \gamma_{j_{k-t} s_{k-t}}\bfA_{j_{k-t} s_{k-t}}
\end{bmatrix}.
\end{equation*}
By Lemma~\ref{lem:cauchy-det}, since the parameters $\{\alpha_{j_1},\ldots,\alpha_{j_{k-t}},\beta_{s_1},\ldots,\beta_{s_{k-t}}\}$ are chosen as distinct elements of the field, we have
\begin{equation*}
\det(\bfA)
=\left(
\frac{
      \displaystyle\prod_{1\le x<y\le k-t}(\alpha_{j_y}-\alpha_{j_x})(\beta_{s_x}-\beta_{s_y})
    }{
      \displaystyle\prod_{x=1}^{k-t}\prod_{y=1}^{k-t}(\alpha_{j_x}-\beta_{s_y})
    }\right)^M\neq 0 \,.
\end{equation*}
This implies that the determinant of the full $k\times k$ block submatrix is
$\det(\bfI_{tM\times tM})\cdot \det(\bfA)=\det(\bfA)\neq 0$.
Thus, every $k\times k$ block submatrix of $\bfG^\top$ is full rank.
\end{proof}

To illustrate Theorem~\ref{thm:main}, we give an explicit generator matrix for a $(7,4,4)$ array code in Example~\ref{ex:1} and verify that it satisfies the MDS property.

\begin{example}\label{ex:2}
Consider the $(7,4,4)$ array code $\cC$ in Example~\ref{ex:1}. Let $\bfm=(\bfa^{(1)},\bfa^{(2)},\bfa^{(3)},\bfa^{(4)})^\top$ be the information column. The array codeword matrix $\bfC\in\bbF_8^{7\times 4}$ is obtained from
$\bfC=\bfG^\top \bfm$,
where $\bfG^\top$ is a $28\times 16$ generator matrix and the output is reshaped as a $7\times 4$ array with one row per storage node. Applying Theorem~\ref{thm:main}, we consider $\bbF_8=\{0,1,\omega,\omega^2,\omega^3,\omega^4,\omega^5,\omega^6\}$, where $\omega$ is a root of $x^3+x+1$. Set $\alpha_1=0$, $\alpha_2=1$, $\alpha_3=\omega$, and $\beta_1=\omega^2$, $\beta_2=\omega^3$, $\beta_3=\omega^4$, $\beta_4=\omega^5$. We then define the Cauchy coefficients $\gamma_{ij}=(\alpha_i-\beta_j)^{-1}$ for $i\in[3]$, $j\in[4]$. With these choices, the explicit generator matrix $\bfG^\top$ becomes
\begin{equation*}
\bfG^\top =
\begin{bmatrix}
\bfI_4 & \bf0   & \bf0   & \bf0 \\
\bf0   & \bfI_4 & \bf0   & \bf0 \\
\bf0   & \bf0   & \bfI_4 & \bf0 \\
\bf0   & \bf0   & \bf0   & \bfI_4 \\
\hline
\gamma_{11}\bfA_{11} & \gamma_{12}\bfA_{12} & \gamma_{13}\bfA_{13} & \gamma_{14}\bfA_{14} \\
\gamma_{21}\bfA_{21} & \gamma_{22}\bfA_{22} & \gamma_{23}\bfA_{23} & \gamma_{24}\bfA_{24} \\
\gamma_{31}\bfA_{31} & \gamma_{32}\bfA_{32} & \gamma_{33}\bfA_{33} & \gamma_{34}\bfA_{34}
\end{bmatrix}
=
\begin{bmatrix}
\bfI_4 & \bf0   & \bf0   & \bf0 \\
\bf0   & \bfI_4 & \bf0   & \bf0 \\
\bf0   & \bf0   & \bfI_4 & \bf0 \\
\bf0   & \bf0   & \bf0   & \bfI_4 \\
\hline
\omega^5\bfA_{11} & \omega^4\bfA_{12} & \omega^3\bfA_{13} & \omega^2\bfA_{14} \\
\omega   \bfA_{21} & \omega^6\bfA_{22} & \omega^2\bfA_{23} & \omega^3\bfA_{24} \\
\omega^3\bfA_{31} & \bfA_{32}          & \omega^5\bfA_{33} & \omega   \bfA_{34}
\end{bmatrix},
\end{equation*}
where each $\bfA_{ij}$ is a $4\times 4$ permutation matrix by the permutation labels. In this case, we have $\bfA_{11}=\bfA_{12}=\bfA_{13}=\bfA_{14}=\bfA_{23}=\bfA_{24}=\bfA_{31}=\bfA_{32}=\bfI_4$. The remaining nontrivial permutation blocks are
\[
\bfA_{21}=\bfA_{33}=
\begin{bmatrix}
0&0&0&1\\
0&0&1&0\\
0&1&0&0\\
1&0&0&0
\end{bmatrix},
\qquad
\bfA_{22}=\bfA_{34}=
\begin{bmatrix}
0&1&0&0\\
1&0&0&0\\
0&0&0&1\\
0&0&1&0
\end{bmatrix}.
\]
With this choice, a direct check in SageMath confirms that any $4$ out of $7$ rows of $\bfG^\top$ yield a $4\times 4$ block matrix (i.e., $16\times 16$) of full rank, hence the array code is MDS.
\end{example}

\section{Zigzag codes with Low Skip Cost}
\label{sec:skipcost}

We now take stock of the development so far. 
In Section~\ref{sec:zigzag}, we recalled the zigzag framework of Tamo–Wang–Bruck and explained why it achieves optimal access and rebuilding ratio. 
In Section~\ref{sec:MDSproperty}, we established that when the row-index group is $\cG=\bbZ_2^m$, one can assign explicit MDS coefficients over a field of linear size. 
Having established these properties, this section focuses on further reducing access latency by ensuring that reads from each helper node are contiguous.

Compared with our earlier work~\cite{chee2024repairing}, our construction of zigzag codes with zero skip cost in this work achieves two notable improvements. 
First, the asymptotic code rate increases from $1/2$ to $2/3$. 
Second, the information length $k$ and the subpacketization $M=2^m$ can now be chosen independently, whereas in~\cite{chee2024repairing} the subpacketization $M$ was dependent on $k$. 

To obtain rates exceeding $2/3$, we incur a small sacrifice in skip cost. 
As we shall see in explicit constructions later in this section, even with this tradeoff, the resulting skip cost remains modest, and the achieved code rates are comparable to those of MDS array codes deployed in practice. A summary is given in Table~\ref{tab:parameters}.

Before presenting the detailed constructions, we outline the general strategy at a high level. 
The main idea is to fix an ordering $\tau$ of the row-index group $\cG$, select a small collection of index-two subgroups $H_1,\ldots,H_t$, and then assign permutation labels to parity nodes so that the repair of any systematic node involves downloading from a single coset of some $H_i$. 
The skip cost is then controlled by the choice of $\tau$ and the structure of the $H_i$. These choices, through total reads, then determine the overall RFR of the repair scheme. We formalize this in the following definition and lemma.

\begin{definition}
Let $\cG$ be a row-index group of size $M$ together with an ordering $\tau$, so that $\cG = \{g_i\}_{i\in[M]}$. 
For a subset $H = \{g_{i_1}, \ldots, g_{i_t}\}$ with $i_1 < \cdots < i_t$, we define the \emph{skip cost of $H$ with respect to $\tau$} as
\[
\cost(H;\tau) \triangleq i_t - i_1 - (t-1).
\]  
If $H$ is a subgroup of index two, we further define
\[
c(H;\tau) \triangleq \min\{\,\cost(H;\tau),\; \cost(\cG \setminus H;\tau)\,\}.
\]
\end{definition}

\begin{lemma}\label{lem:construction}
Let $t,s \ge 1$ and let $\cG$ be a row-index group of size $M$. 
Suppose $\tau$ is an ordering of $\cG$, and let $H_1,\ldots,H_t$ be index-two subgroups of $\cG$. 
Assume that for every $i \in [t]$:
\begin{enumerate}[(T1)]
    \item there exists $\bfh_i \in \Bigl(\bigcap_{j \ne i} H_j\Bigr) \setminus H_i$; and
    \item $c(H_i;\tau) \le \sigma$,
\end{enumerate}
for some integer $\sigma \ge 0$.  
Then there exists a $(N = st+s+1,\; k=st,\; M)$  zigzag code whose skip cost is at most $\sigma(k+1)$.
\end{lemma}

\begin{proof}
Set $p = s+1$ and $k = st$. 
Within the zigzag code framework (Section~\ref{sec:zigzag}), it suffices to specify the $p$ permutation labels 
$\bfS_0,\bfS_1,\ldots,\bfS_s \in \cG^{\,k}$ for the parity nodes and then give, for each failed systematic node, a repair scheme that downloads $M/2$ symbols per helper, with the desired skip cost.

\noindent{\em Permutation labels}. 
Set $\bfS_0 = (\vzero,\ldots,\vzero)$. 
For the other parities, we partition the $k$ systematic indices into $s$ consecutive blocks of length $t$:
$B_j = \{(j-1)t+1,\,(j-1)t+2,\,\ldots,\,jt\}$  for $j\in[s]$.
We also fix the elements $\bfh_1,\ldots,\bfh_t$ from (T1).
Then for $j\in[s]$, we place $(\bfh_1,\ldots,\bfh_t)$ in the block $B_j$ of $\bfS_j$ and $\vzero$ outside $B_j$. In other words, $\bfS_1 = (\bfh_1,\ldots, \bfh_t,\vzero,\ldots,\vzero)$ and subsequent $\bfS_j$'s are cyclic shifts of $\bfS_1$ by $(j-1)t$ positions.

\noindent{\em Repair scheme}.
Fix a failed systematic node whose index is the $i$-th position inside block $j$; we denote it by $(i,j)$ with 
$(i,j)\in [t]\times[s]$. 
Let $H_i^*$ be whichever of $H_i$ or its coset $\cG\setminus H_i$ that attains $c(H_i;\tau)$. So, $|H_i^*|=M/2$ and $\cost(H_i^*;\tau)=c(H_i;\tau)$.
Then to repair the systematic node $(i,j)$, we contact the helper set
\[
\cH \bigl((i,j)\bigr) \;=\; 
\bigl\{\bfa^{(i',j')} \; : \; (i',j')\ne (i,j) \bigr \}\;\cup\;
\bigl\{\,\bfp^{(0)}, \bfp^{(j)}\,\bigr\},
\]
and each helper in $\cH((i,j))$ returns exactly the symbols in rows indexed by $H_i^*$ (i.e., $M/2$ symbols).

\begin{itemize}
\item {\em Correctness}. We apply the Recovery Lemma (Lemma~\ref{lem:recovery}) and look at the labels $\bfS_0$ and $\bfS_j$.
First, at all positions outside $(i,j)$, we either have $\vzero$ or $\bfh_{i'}$ with $i'\ne i$. In both cases, $\vzero$ and $\bfh_{i'}$ belong to $H_i$ (since $i\ne i'$ implies that $\bfh_{i'}\in H_i$ by Condition (T1)) and so, Condition (R2) in Lemma~\ref{lem:recovery} is met. 
On the other hand, at position $(i,j)$, the corresponding element is $\bfh_i$ which does not belong to $H_i$ and so, $\{\vzero, \bfh_i\}$ is a complete set of coset representatives. Therefore, Condition (R1) is met and downloading the rows in $H_i^*$ from the helpers in $\cH((i,j))$ enables recovery.
\item {\em Skip cost}. Every helper (the $k-1$ systematic nodes and the two parities) reads from the same row set $H_i^*$; hence each helper incurs skip cost $\cost(H_i^*;\tau)=c(H_i;\tau)\le \sigma$ (Condition (T2)).
There are $(k-1)+2=k+1$ helpers, so the total skip cost is at most $(k+1)\sigma$.
\end{itemize}

Therefore, the resulting zigzag code has parameters $N=st+s+1$, $k=st$, $M$ and admits a repair scheme with total skip cost at most $\sigma(k+1)$, as claimed.
\end{proof}

To apply Lemma~\ref{lem:construction}, it remains to identify an ordering $\tau$ of $\cG$ together with index-two subgroups $H_1,\ldots,H_t$ such that $c(H_i;\tau)$ are small. 
When $t=2$, we find that the lexicographic ordering already suffices: there exist two subgroups $H_1,H_2$ with $c(H_1;\tau)=c(H_2;\tau)=0$. 
Hence, we obtain zero-skip cost zigzag codes whose rates approach $2/3$ (see Construction~A(i)).

For $t \ge 3$, the situation changes. 
We performed an exhaustive search for $M \in \{4,8\}$ and were unable to find any ordering $\tau$ and subgroups $H_1,H_2,H_3$ of $\bbZ_2^2$ or $\bbZ_2^3$ such that $c(H_i;\tau)=0$ for all $i$. 
We believe this is impossible more generally, which leads us to the following conjecture.

\begin{conjecture}
For every $t \ge 3$ and every $M=2^m$, there does not exist an ordering $\tau$ of $\bbZ_2^m$ together with index-two subgroups $H_1,\ldots,H_t$ such that $c(H_i;\tau)=0$ for all $i \in [t]$.
\end{conjecture}

Nevertheless, for $t \in \{3,4\}$ and $M \in \{8,16\}$, we carried out a search and obtained explicit families where each $c(H_i;\tau)$ is small (though nonzero). 
These lead to constructions with bounded skip cost, which we present in the following theorem.

\begin{theorem}[Construction A]
Fix $s \ge 1$ and consider $\cG = \bbZ_2^m$ with $M=2^m$. 
With suitable choices of the ordering $\tau$ and subgroups $H_i$, the following families of zigzag codes exist:
\begin{enumerate}[(i)]
    \item \textbf{Case $t=2$, $m \ge 2$.}  
    Let $\tau$ be the usual lexicographic order.  
    Define
    \[
    H_1 = \{x_1=0\}, \qquad H_2 = \{x_1=x_2\}.
    \]
    Then there exists a $(3s+1,\,2s,\,M)$ zigzag code with zero skip cost.  
    As $s \to \infty$, the rate approaches $2/3$.
    
    \item \textbf{Case $t=3$, $m=3$ ($M=8$).}  
    Let $\tau$ be the order $000,001,010,011,100,110,101,111$.
    Define
    \[
    H_1 = \{x_1=0\}, \qquad 
    H_2 = \{x_1=x_3\}, \qquad
    H_3 = \{x_1=x_2\}.
    \]
    Then there exists a $(4s+1,\,3s,\,M)$ zigzag code whose skip cost is at most $(k+1)$.  
    As $s \to \infty$, the rate approaches $3/4$.

        \item \textbf{Case $t=4$, $m=4$ ($M=16$).}  
    Let $\tau$ be the row order
    $1000,\,1010,\,1100,\,1110,\,0101,\,0011,\,1111,\,1011,0100,\,0000,\,0001$, $0010,\,1101,\,0110,\,0111,\,1001.$
    Define
    \[
    H_1=\{ x_1 = x_2 \}, \quad
    H_2=\{ x_1 = x_3 \}, \quad
    H_3=\{ x_1 = x_4\}, \quad
    H_4=\{ x_1 = 0\}.
    \]
    Then there exists a $(5s+1,\,4s,\,M)$ zigzag code with skip cost at most $3(k+1)$.  
    As $s \to \infty$, the rate approaches $4/5$.

\end{enumerate}
\end{theorem}

\begin{proof}
(i) Choose $\vone \in H_2 \setminus H_1$ and $\bfe_2 \in H_1 \setminus H_2$, so that (T1) holds.  
Under lexicographic order, $H_1$ and the coset $\cG\setminus H_2$ form contiguous blocks, hence
\[
c(H_1;\tau) = c(H_2;\tau) = 0.
\]
Thus, Lemma~\ref{lem:construction} applies with $\sigma=0$, yielding the claimed code parameters.

(ii) Choose $\bfh_1=(1,1,1)$, $\bfh_2=(0,0,1)$, $\bfh_3=(0,1,0)$, which satisfy (T1).  
Under the specified order $\tau$, we check that
\begin{align*}
    c(H_1;\tau) & =\cost(H_1;\tau) = 0,\\
    c(H_2;\tau) & =\cost(\cG\setminus H_2;\tau) = 1,\\
    c(H_3;\tau) & =\cost(\cG\setminus H_3;\tau) = 1.
\end{align*}
Hence, Lemma~\ref{lem:construction}  applies with $\sigma=1$, and the total skip cost is at most $(k+1)$. 

(iii) Choose $\bfh_1=(0,1,0,0)$, $\bfh_2=(0,0,1,0)$, $\bfh_3=(0,0,0,1)$, and $\bfh_4=(1,1,1,1)$, which satisfy (T1).  
Under the specified order $\tau$, we check that
\begin{align*}
    c(H_1;\tau) & =\cost( H_1;\tau) = 3,\\
    c(H_2;\tau) & =\cost( H_2;\tau) = 2,\\
    c(H_3;\tau) & =\cost( H_3;\tau) = 2,\\
    c(H_4;\tau) & =\cost( H_4;\tau) = 3.
\end{align*}
Hence, Lemma~\ref{lem:construction}  applies with $\sigma=3$, and the total skip cost is at most $3(k+1)$. 
\end{proof}

\begin{example}[Example~\ref{ex:1} continued]
Setting $t=2$, $s=2$, and $M=4$ in Construction~A(i), we obtain the permutation labels of the MDS array code in Example~\ref{ex:1}. 
In the repair scheme for each failed systematic node, the accessed rows are always drawn from either the subgroup $H_1=\{00,01\}$ or the coset $\cG \setminus H_2=\{01,10\}$, both of which form consecutive row blocks under the lexicographic ordering of $\cG$. 
Moreover, each systematic node is repaired by accessing exactly half of the rows from the remaining systematic nodes and the two parity nodes, achieving the optimal rebuilding ratio of $1/2$. 
Recall also from Example~\ref{ex:2} that over the field $\bbF_8$, one can choose coefficients that result in MDS property.
\end{example}

By shortening the zigzag codes obtained in Construction~A, that is, by removing certain systematic columns while keeping the parity structure unchanged, we obtain further families of zigzag codes. 
These shortened codes inherit the repair properties of the original construction, and the skip costs follow directly.

\begin{corollary}
Fix $s \ge 1$.
\begin{enumerate}[(i)]
    \item For $M \ge 4$ and $2s-1 \le k \le 2s$, there exists a $(k+s+1,\,k,\,M)$ zigzag code with zero skip cost.
    \item For $M = 8$ and $3s-2 \le k \le 3s$, there exists a $(k+s+1,\,k,\,M)$ zigzag code with skip cost at most $(k+1)$.
    \item For $M = 16$ and $4s-3 \le k \le 4s$, there exists a $(k+s+1,\,k,\,M)$ zigzag code with skip cost at most $3(k+1)$.
\end{enumerate}
\end{corollary}

In the next construction, we restrict attention to the case of two parity nodes. 
For $M=8$ and $M=16$, this improves the achievable rates from $3/5$ to $4/6$ and from $4/6$ to $5/7$, respectively. 

\begin{theorem}[Construction B - Two Parities]
Consider $p=2$. In other words, $N=k+2$.
\begin{enumerate}[(i)]
\item Let $\cG=\bbZ_2^3$ with $M=8$ and the row order is
$001,010,011,000,100,101,110,111$.
Consider the two parity nodes with permutation labels
\[
\bfS_1=(\vzero,\vzero,\vzero,\vzero), \qquad 
\bfS_2=(000,100,110,101).
\]

Then this defines a $(N,k,M)=(6,4,8)$ zigzag code with skip cost $10$.

    \item Let $\cG=\bbZ_2^4$ with $M=16$ and row order 
$    1000,1101,1110,1111,1100,1011,0000,0001$, $0010,0011,0100,0101,0110$, $1001,1010,0111$.
    Consider the two parity nodes with permutation labels
    \[
    \bfS_1=(\vzero,\vzero,\vzero,\vzero,\vzero), \qquad 
    \bfS_2=(0000,0100,0010,0001,1111).
    \]
    Then this defines a $(N,k,M)=(7,5,16)$ zigzag code with maximum skip cost $30$.
\end{enumerate}

\end{theorem}

\begin{proof}
(i) We specify the repair scheme for each systematic node $\bfa^{(i)}$ with $i\in [4]$.
Define the subgroups
\[
H_1=\{x_1=0\},\quad 
H_2=\{x_1+x_2+x_3=0\},\quad 
H_3=\{x_2=0\},\quad 
H_4=\{x_3=0\}.
\]

\begin{itemize}
    \item \emph{Repair of $\bfa^{(1)}$.}  
    Download rows indexed by $H_1=\{001,010,011,000\}$ from $\{\bfa^{(i)}:i\in\{2,3,4\}\}\cup\{\bfp^{(1)}\}$,  
    and rows indexed by $\cG\setminus H_1$ from $\bfp^{(2)}$.  
    Each helper incurs skip cost $0$, so the total skip cost is $0\cdot 5=0$. Note that for this case, the recovery works because each of $(000,100)$, $(000,110)$, and $(000,101)$ forms a complete coset representatives of $H_1$, while $(000,000)\in H_1$, ensuring that Conditions (R1$'$) and (R2$'$) of Lemma~\ref{lem:recovery} are satisfied.  

    \item \emph{Repair of $\bfa^{(2)}$.}  
    Download rows indexed by $\cG\setminus H_2=\{011,000,101,110\}$ from $\{\bfa^{(i)}:i\in\{1,3,4\}\}\cup\{\bfp^{(1)},\bfp^{(2)}\}$.  
    Each helper incurs skip cost $1$, giving total skip cost $1\cdot 5=5$.

    \item \emph{Repair of $\bfa^{(3)}$.}  
    Download rows indexed by $H_3=\{001,000,100,101\}$ from $\{\bfa^{(i)}:i\in\{1,2,4\}\}\cup\{\bfp^{(1)},\bfp^{(2)}\}$.  
    Each helper incurs skip cost $2$, giving total skip cost $2\cdot 5=10$.

    \item \emph{Repair of $\bfa^{(4)}$.}  
    Download rows indexed by $H_4=\{010,000,100,110\}$ from $\{\bfa^{(i)}:i\in\{1,2,3\}\}\cup\{\bfp^{(1)},\bfp^{(2)}\}$.  
    Each helper incurs skip cost $2$, giving total skip cost $2\cdot 5=10$.
\end{itemize}

Thus, the maximum skip cost is $10$.

(ii) We specify the repair scheme for each systematic node $\bfa^{(i)}$ with $i\in[5]$.
Define
    \[
    H_1=\{x_2+x_3+x_4=0\}, \quad 
    H_2=\{x_1+x_2=0\}, \quad 
    H_3=\{x_1+x_3=0\}, \quad 
    H_4=\{x_1+x_4=0\}, \quad 
    H_5=\{x_1=0\}.
    \]
\begin{itemize}
    \item \emph{Repair of $\bfa^{(1)}$.}  
    Download rows indexed by $H_1=\{1000,1101,1110,1011,0000,0011,0101,0110\}$ 
    from $\{\bfa^{(i)}:i\in\{2,3,4,5\}\}\cup\{\bfp^{(1)}\}$,  
    and rows indexed by $\cG\setminus H_1$ from $\bfp^{(2)}$.  
    Each helper incurs skip cost $5$, giving total $5\cdot 6=30$.

    \item \emph{Repair of $\bfa^{(2)}$.}  
    Download rows indexed by $H_2=\{1101,1110,1111,1100,0000,0001,0010,0011\}$ 
    from $\{\bfa^{(i)}:i\in\{1,3,4,5\}\}\cup\{\bfp^{(1)},\bfp^{(2)}\}$.  
    Each helper incurs skip cost $1$, giving total $1\cdot 6=6$.

    \item \emph{Repair of $\bfa^{(3)}$.}  
    Download rows indexed by $H_3=\{1110,1111,1011,0000,0001,0100,0101,1010\}$ 
    from $\{\bfa^{(i)}:i\in\{1,2,4,5\}\}\cup\{\bfp^{(1)},\bfp^{(2)}\}$.  
    Each helper incurs skip cost $5$, giving total $5\cdot 6=30$.

    \item \emph{Repair of $\bfa^{(4)}$.}  
    Download rows indexed by $H_4=\{1101,1111,1011,0000,0010,0100,0110,1001\}$ 
    from $\{\bfa^{(i)}:i\in\{1,2,3,5\}\}\cup\{\bfp^{(1)},\bfp^{(2)}\}$.  
    Each helper incurs skip cost $5$, giving total $5\cdot 6=30$.

    \item \emph{Repair of $\bfa^{(5)}$.}  
    Download rows indexed by $H_5=\{0000,0001,0010,0011,0100,0101,0110,0111\}$ 
    from $\{\bfa^{(i)}:i\in\{1,2,3,4\}\}\cup\{\bfp^{(1)},\bfp^{(2)}\}$.  
    Each helper incurs skip cost $2$, giving total $2\cdot 6=12$.
\end{itemize}

Thus, the maximum skip cost across all repairs is $30$.

\end{proof}

\begin{remark}
Table~\ref{tab:parameters} provides an illustrative summary of our results by comparing the parameters of widely deployed Reed–Solomon codes in storage systems (see for example~\cite{dinh2022}) with the corresponding zigzag codes that achieve small skip cost. 
The table highlights that in many practical settings, zigzag codes can match or even improve upon the $(N,k)$ parameters of RS codes, while offering bounded skip cost that enables more efficient repair. 
In particular, our constructions yield codes with zero skip cost (e.g., for Filebase/Sia and Storj), as well as low skip cost multiples of $(k+1)$ for common configurations such as $(9,6)$ and $(11,8)$. 
This demonstrates that zigzag codes provide a viable alternative to RS codes for distributed storage systems, combining the MDS property with improved repair bandwidth and access latency.

\afterpage{
\clearpage
 \begin{sidewaystable}[htbp]
 \renewcommand{\arraystretch}{1.3}
\centering
\small
\begin{tabular}{|C{3.5cm}|C{1cm}|C{1.5cm}|C{2cm}||C{2cm}|C{2cm}|C{2cm}|C{2cm}|C{4cm}| }
\hline
\multicolumn{4}{|c||}{\small\textbf{RS codes used in practice}\rule{0pt}{3ex}} &
\multicolumn{5}{c|}{\small\textbf{Zigzag codes with small skip cost}\rule{0pt}{3ex}} \\

\hline
\textbf{System} & \textbf{RS $(N,k)$} &   \textbf{rebuilding ratio} & \textbf{repair-fragmentation ratio} & \textbf{Zigzag $(N,k,M)$} & \textbf{Skip cost} & \textbf{rebuilding ratio} & \textbf{repair fragmentation ratio} & \textbf{Construction (with $s$)} \\
\hline
IBM Spectrum Scale RAID     &  \multirow{2}{*}{$(10,8)$}  & \multirow{2}{*}{$1$} & \multirow{2}{*}{$1$}
  &  \multirow{2}{*}{\textcolor{red}{$(11,8,16)$}} &  \multirow{2}{*} {$3(k{+}1)=27$} & \multirow{2}{*}{$1/2$} & \multirow{2}{*}{$11/16$} &\multirow{2}{*}{Construction~A(iii), $s=2$ }\\
Linux RAID-6               &  &  &  & & 
  &  &  & \\ \hline
Google FS II              & \multirow{3}{*}{(9,6)} & \multirow{3}{*}{1}& \multirow{3}{*}{1} & \multirow{3}{*}{(9,6,8)} &  \multirow{3}{*}{$(k+1)=7$} &  \multirow{3}{*}{$1/2$} &  \multirow{3}{*}{$5/8$} &  \multirow{3}{*}{Construction~A(ii), $s=2$}  \\
Quantcast File System    &  & &  &  &  &  & &  \\
Hadoop 3.0 HDFS-EC       &  & &  &  &  &  & &  \\ \hline
Yahoo Cloud Object Store    & $(11,8)$ & $1$ & $1$
  & $(11,8,16)$             & $3(k{+}1)=27$  &  $1/2$ &  $11/16$ & Construction~A(iii), $s=2$ \\ \hline
Backblaze Vaults (online backup) & $(20,17)$ & $1$ & $1$ 
  & \textcolor{red}{$(23,17,16)$} & $3(k+1)=54$            &   $1/2$ & $11/16$ & Shortening Construction~A(iii), $s=5$ \\ \hline
Facebook’s f4 storage       & $(14,10)$  & $1$ & $1$ 
  & $(14,10,16)$              & $3(k{+}1)=33$   &   $1/2$ & $11/16$ & Shortening Construction~A(iii), $s=3$ \\ \hline
Baidu Atlas Cloud Storage   & $(12,8)$  & $1$ & $1$
  & $(12,8,8)$                & $(k{+}1)=9$    &   $1/2$ & $5/8$ & Shortening Construction~A(ii), $s=3$ \\ \hline
Microsoft Pelican (cold)    & $(18,15)$ & $1$ & $1$
  & $(20,15,16)$ & $3(k+1)=48$            &  $1/2$ & $11/16$ & Shortening Construction~A(iii), $s=4$ \\ \hline
Filebase/Sia (blockchain)   & $(30,10)$ & $1$ & $1$
  & $(16,10,4)$               & $0$            &  $1/2$ & $1/2$ &Construction~A(i), $s=5$ \\ \hline
\multirow{2}{*}{Storj (blockchain)}  & $(40,20)$  & \multirow{2}{*}{$1$} & \multirow{2}{*}{$1$} &  $(31,20,4)$ &  \multirow{2}{*}{$0$}          &  $1/2$ & $1/2$ &  Construction~A(i), $s=10/20$ \\
& $(80,40)$ &  & 
  & $(61,40,4)$  & $0$       &  $1/2$ &  $1/2$ & Construction~A(i), $s=10/20$ \\\hline
\end{tabular}


\caption{Reed–Solomon codes in practice and zigzag codes from this work. 
Entries in \textcolor{red}{red} have slightly lower rates (more parity) than the listed RS profiles but remain close, while achieving bounded skip cost. 
Skip cost is the total (summed across all helpers) per single-node repair.}
\label{tab:parameters}
 \end{sidewaystable}
   \clearpage}

\end{remark}

\section{Conclusion}\label{sec:conclusion}

We revisited zigzag array codes with an emphasis on the row-indexing group $\cG$. 
We proved that when $\cG=\bbZ_2^m$ and the field has characteristic two, explicit coefficients over any field with $|\bbF|\ge N$ guarantee the MDS property, thereby decoupling the dependence among $p$, $k$, and $M$. 

Building on this result, we introduced an ordering-and-subgroup framework that enables repair-by-transfer schemes with bounded skip cost and low repair-fragmentation ratio (RFR), while still retaining optimal access and optimal rebuilding ratio for single-node repair. 

Several open problems remain:  
\begin{enumerate}[(i)]
    \item {\em Lower bounds on RFR.} What are the tight information-theoretic lower bounds on RFR as a function of $(N,k,M,p)$ and the number of helpers $d$? While RFR is trivially lower bounded by the rebuilding ratio and is achieved by zero–skip-cost codes, whether this is possible for all parameters remains unclear.  
    
    \item {\em Tradeoffs.} Characterize the fundamental tradeoff among rate, RFR, subpacketization, and field size beyond the MSR setting. 
    
    \item {\em Multiple failures.} Extend the zigzag framework—particularly Lemma~\ref{lem:recovery}—to the simultaneous repair of $t>1$ failed nodes with bounded RFR.  
    
    \item {\em Other code families.} Can other classes of array codes be adapted to support repair schemes with provably low RFR?  
\end{enumerate}

\bibliographystyle{IEEEtran}
\bibliography{references}

\phantom{SPACE TO FORCE NEW PAGE}

\appendices

\section{Proof of Lemma~\ref{lem:cauchy-det}}\label{ Appendix}

In this appendix, we present the proof of Lemma~\ref{lem:cauchy-det}. Specifically, we derive the determinant formula for a block matrix formed by permutation matrices whose row-index group is $\bbZ_2^m$.

To this end, we adopt the following convention: for $\bfu\in\bbZ_2^m$, we use $\bfP_{\bfu}^{(m)}$ to denote a permutation matrix whose rows and columns are indexed by elements of $\bbZ_2^m$. Moreover, the entries of $\bfP_{\bfu}^{(m)}$ are given by

\begin{equation}\label{eq:perm-matrix-Z2m}
\bfP_{\bfu}^{(m)}(\bfg,\bfh)
=
\begin{cases}
1, & \text{if } \bfh = \bfg + \bfu,\\[2pt]
0, & \text{otherwise,}
\end{cases}
\qquad \bfg,\bfh,\bfu \in \bbZ_2^m .
\end{equation}

To prepare for the proof, we begin by stating a fundamental property of these permutation matrices, followed by several technical lemmas essential to the argument.

\begin{proposition}\label{prop:perm}
    For any $\bfu_1, \bfu_2 \in \cG = \bbZ_2^m$, the permutation matrices $\bfP_{\bfu}^{(m)}$ as in \eqref{eq:perm-matrix-Z2m} satisfy
    \begin{equation}\label{permu_propo}
        \bfP^{(m)}_{\bfu_1} \, \bfP^{(m)}_{\bfu_2} = \bfP^{(m)}_{\bfu_1 + \bfu_2}.
    \end{equation}
    \noindent Therefore, $\bfP^{(m)}_{\bfu_1}$ and $\bfP^{(m)}_{\bfu_2}$ commute.
    \end{proposition}
    
    \begin{proof}
    By \eqref{eq:perm-matrix-Z2m}, the $(\bfg,\bfh)$-th entry of the product $\bfP_{\bfu_1}^{(m)} \bfP_{\bfu_2}^{(m)}$ is
    \[
    \left(\bfP_{\bfu_1}^{(m)} \bfP_{\bfu_2}^{(m)}\right)(\bfg, \bfh) = \sum_{\bfk\in\bbZ_2^m} \bfP_{\bfu_1}^{(m)}(\bfg, \bfk)\, \bfP_{\bfu_2}^{(m)}(\bfk, \bfh).
    \]

Since $\bfP_{\bfu_1}^{(m)}(\bfg,\bfk) = 1$ if and only if $\bfk=\bfg+\bfu_1$,
 we get 
$\bfP_{\bfu_1}^{(m)} \bfP_{\bfu_2}^{(m)}(\bfg,\bfh)
= \bfP_{\bfu_2}^{(m)} (\bfg+\bfu_1,\bfh)
= 1$  if and only if $\bfh=(\bfg+\bfu_1)+\bfu_2$.
Thus,
\begin{equation}
    \left(\bfP_{\bfu_1}^{(m)} \bfP_{\bfu_2}^{(m)}\right)(\bfg, \bfh) =
    \begin{cases}
    1, & \text{if } \bfh = \bfg + (\bfu_1 + \bfu_2), \\
    0, & \text{otherwise,}
    \end{cases}
\end{equation}
which implies      
$\bfP_{\bfu_1}^{(m)} \bfP_{\bfu_2}^{(m)} = \bfP_{\bfu_1+\bfu_2}^{(m)}$.
 \end{proof}

 Next, define the sets $G_1=\{ \boldsymbol{u}\in \bbZ_2^m |~u_{1}=0 \}$,  and $G_2=\{ \boldsymbol{u}\in \bbZ_2^m  |~u_{1}=1 \}$. Clearly, $G_1 \cup G_2 = \bbZ_2^m$. We then have the following lemma.

\begin{lemma}\label{lem:perm_recursive}
     Let $m$ be a positive integer.  Let $\bfu = (u_1, \bfu') \in \bbZ_2^m$ with $u_1 \in \bbZ_2$ and $\bfu' = (u_{2}, \ldots, u_m) \in \bbZ_2^{m-1}$. Let $\bfP^{(m)}_{\bfu}$ be the $2^m\times 2^m$ permutation matrix defined in \eqref{eq:perm-matrix-Z2m}. Then
 \begin{equation}\label{le7:eq}
     \bfP^{(m)}_{\bfu} =
     \begin{cases}
     \begin{pmatrix}
     \boldsymbol{P}^{(m-1)}_{\bfu'} & \mathbf{0} \\
     \mathbf{0} & \boldsymbol{P}^{(m-1)}_{\bfu'}
     \end{pmatrix}, & \text{if } u_1 = 0, \\[2em]
     \begin{pmatrix}
         \mathbf{0} & \boldsymbol{P}^{(m-1)}_{\bfu'} \\
     \boldsymbol{P}^{(m-1)}_{\bfu'} &  \mathbf{0}
     \end{pmatrix}, & \text{if } u_1 = 1\,.
     \end{cases}
 \end{equation}
     \end{lemma}
 
     \begin{proof}

By \eqref{eq:perm-matrix-Z2m}, for $\bfu=(u_1,\bfu')\in\bbZ_2^m$ and
$\bfi=(i_1,\bfi')$, $\bfj=(j_1,\bfj')$ with $i_1,j_1\in\bbZ_2$ and
$\bfi',\bfj'\in\bbZ_2^{m-1}$, we have
\[
\bfP^{(m)}_{\bfu}\bigl((i_1,\bfi'),\,(j_1,\bfj')\bigr)=1
\text{ if and only if }
(j_1,\bfj')=(i_1+u_1,\;\bfi'+\bfu').
\]
Thus:
\begin{itemize}
\item If $u_1=0$, then $j_1=i_1$ and $\bfj'=\bfi'+\bfu'$, so the nonzeros lie on the two
diagonal blocks and each such block equals $\bfP^{(m-1)}_{\bfu'}$. Therefore
\[
\bfP^{(m)}_{\bfu}=
\begin{pmatrix}
\bfP^{(m-1)}_{\bfu'} & \mathbf{0}\\
\mathbf{0} & \bfP^{(m-1)}_{\bfu'}
\end{pmatrix}.
\]
\item If $u_1=1$, then $j_1=i_1+1$ and $\bfj'=\bfi'+\bfu'$, so the nonzeros lie on the two
off-diagonal blocks and each such block equals $\bfP^{(m-1)}_{\bfu'}$. Therefore
\[
\bfP^{(m)}_{\bfu}=
\begin{pmatrix}
\mathbf{0} & \bfP^{(m-1)}_{\bfu'}\\
\bfP^{(m-1)}_{\bfu'} & \mathbf{0}
\end{pmatrix}.
\]
\end{itemize}
Combining the two cases yields \eqref{le7:eq}.
\end{proof}
 
To facilitate the computation of determinants of block matrices whose entries commute, we recall a fundamental result from~\cite{kovacs1999determinants}. 

\begin{theorem}{\cite{kovacs1999determinants}}\label{thm:Laplace}
Let $\cR$ be a commutative ring. Suppose $\bfG$ is a $k\times k$ block matrix with blocks
$A_{i,j}\in \cR^{\,n\times n}$ that commute pairwise. Then
\[
\det(\bfG)
=
\det\!\Bigg(
   \sum_{\pi\in S_k} (\operatorname{sgn}\pi)\,
   A_{1,\pi(1)} A_{2,\pi(2)} \cdots A_{k,\pi(k)}
\Bigg).
\]
Here, $S_k$ is the symmetric group on the set $ \{1,2,\ldots,k\} $, and $\mathrm{sgn}(\pi)$ denotes the sign of the permutation $\pi\in S_k$.
\end{theorem}

\noindent In fields with \textit{characteristic two,} we have that \begin{equation}\label{eq:laplace-char2}
\det(\bfG)
=
\det\!\Bigg(
   \sum_{\pi\in S_k}
   A_{1,\pi(1)} A_{2,\pi(2)} \cdots A_{k,\pi(k)}
\Bigg).
\end{equation}

Based on \eqref{eq:perm-matrix-Z2m}, Lemma~\ref{lem:perm_recursive}, and Theorem~\ref{thm:Laplace}, we have the following lemma.

\begin{lemma}\label{thm:perm_determinant}
Let $m\ge 1$, and let $\{\alpha_{\bfu}\}_{\bfu\in\bbZ_2^m}$ be a set of coefficients in a field $\bbF$ with characteristic two. If
\[
A^{(m)} \;=\; \sum_{\bfu\in\bbZ_2^m} \alpha_{\bfu}\,\bfP^{(m)}_{\bfu}\,,
\]
then
\[
\det\!\big(A^{(m)}\big) \;=\; \left(\,\sum_{\bfu\in\bbZ_2^m}\alpha_{\bfu}\,\right)^{2^{m}}.
\]
\end{lemma}

\begin{proof}
We proceed by induction on  $m$. For the base case $m=1$, we have $M=2$ and $\bbZ_{2}=\{0,1\}$.  The corresponding permutation matrices are
     \begin{equation*}
           \bfP_0^{(1)} = 
           \begin{pmatrix}
           1 & 0 \\ 0 & 1
           \end{pmatrix}, \qquad
           \bfP_1^{(1)} = 
           \begin{pmatrix}
           0 & 1 \\ 1 & 0
           \end{pmatrix}.
         \end{equation*}
         Thus,
         \[
         A^{(1)} = \alpha_0 \bfP_0^{(1)} + \alpha_1 \bfP_1^{(1)} =
         \begin{pmatrix}
             \alpha_0 & \alpha_1 \\
             \alpha_1 & \alpha_0
         \end{pmatrix},
         \]
whose determinant is
         $\det(A^{(1)}) =\alpha_0^2 - \alpha_1^2 = (\alpha_0 + \alpha_1)^2$,  confirming the result for $m=1$.

Assume the claim holds for $m-1$ and we consider $A^{(m)} = \sum_{\bfu\in\bbZ_2^m} \alpha_{\bfu}\,\bfP^{(m)}_{\bfu}$. 

Write each $\bfu\in\bbZ_2^m$ as $(u_1,\bfu')$ with $\bfu'\in\bbZ_2^{m-1}$. By Lemma~\ref{lem:perm_recursive}, we can rewrite the matrix as the sum 
\[
A^{(m)}
=\sum_{\bfu'\in\bbZ_2^{m-1}}\!\alpha_{(0,\bfu')}\!
\begin{pmatrix}\bfP^{(m-1)}_{\bfu'}&\mathbf{0}\\ \mathbf{0}&\bfP^{(m-1)}_{\bfu'}\end{pmatrix}
+\sum_{\bfu'\in\bbZ_2^{m-1}}\!\alpha_{(1,\bfu')}\!
\begin{pmatrix}\mathbf{0}&\bfP^{(m-1)}_{\bfu'}\\ \bfP^{(m-1)}_{\bfu'}&\mathbf{0}\end{pmatrix}.
\]
Hence, we have that
\[
A^{(m)}=
\begin{pmatrix}
A^{(m-1)}_{0} & A^{(m-1)}_{1}\\
A^{(m-1)}_{1} & A^{(m-1)}_{0}
\end{pmatrix},
\text{ where }
A^{(m-1)}_{0}\triangleq \sum_{\bfu'}\alpha_{(0,\bfu')}\bfP^{(m-1)}_{\bfu'},
\text{ and }
A^{(m-1)}_{1}\triangleq\sum_{\bfu'}\alpha_{(1,\bfu')}\bfP^{(m-1)}_{\bfu'}.
\]
Since the $\bfP^{(m-1)}_{\bfu'}$ commute pairwise, $A^{(m-1)}_{0}$ and $A^{(m-1)}_{1}$ commute. 
Then Theorem~\ref{thm:Laplace} implies that
\[
\det\!\big(A^{(m)}\big)
=\det\!\big((A^{(m-1)}_{0})^2+(A^{(m-1)}_{1})^2\big)
= \det\!\big(A^{(m-1)}_{0}+A^{(m-1)}_{1}\big)^{2}.
\]
Notice that 
\[
A^{(m-1)}_{0}+A^{(m-1)}_{1}
=\sum_{\bfu'\in\bbZ_2^{m-1}}\!\big(\alpha_{(0,\bfu')}+\alpha_{(1,\bfu')}\big)\,\bfP^{(m-1)}_{\bfu'},
\]
so by the induction hypothesis,
\[
\det\!\big(A^{(m-1)}_{0}+A^{(m-1)}_{1}\big)
=\Big(\,\sum_{\bfu'\in\bbZ_2^{m-1}}\big(\alpha_{(0,\bfu')}+\alpha_{(1,\bfu')}\big)\,\Big)^{2^{m-1}}
=\Big(\,\sum_{\bfu\in\bbZ_2^{m}}\alpha_{\bfu}\,\Big)^{2^{m-1}}.
\]
Squaring gives $\det(A^{(m)})=\big(\sum_{\bfu}\alpha_{\bfu}\big)^{2^{m}}$, as claimed.
     \end{proof}

\begin{remark}
Lemma~\ref{thm:perm_determinant} also applies when the sum runs over any subset $S\subseteq\bbZ_2^m$.
Specifically, if $ A^{(m)}_S=\sum_{\bfu\in S}\alpha_{\bfu}\,\bfP^{(m)}_{\bfu}$, then $
\det\!\big(A^{(m)}_S\big)=\Big(\sum_{\bfu\in S}\alpha_{\bfu}\Big)^{2^m}$.
This is because the sum can be generalized to all $\bfu \in \bbZ_2^m$ by setting $\alpha_{\bfu} = 0$ for all $\bfu \notin S$, without affecting the value of the determinant.
\end{remark}

By combining Theorem~\ref{thm:Laplace}, Lemma~\ref{lem:perm_recursive}, and Proposition~\ref{prop:perm}, we give the proof of Lemma~\ref{lem:cauchy-det}.

\noindent\textbf{Lemma \ref{lem:cauchy-det}.}
Let $M,N\ge 1$, and let $\bbF$ be a finite field of characteristic two.
Let $\alpha_1,\ldots,\alpha_N$ and $\beta_1,\ldots,\beta_N$ be $2N$ distinct elements of $\bbF$, and define $\gamma_{ij}=(\alpha_i-\beta_j)^{-1}$ for $1\le i,j\le N$.
Let $\bfA$ be the $MN\times MN$ block matrix
\[
\bfA =
\begin{bmatrix}
\gamma_{11} A_{11} & \gamma_{12} A_{12} & \cdots & \gamma_{1N} A_{1N} \\
\gamma_{21} A_{21} & \gamma_{22} A_{22} & \cdots & \gamma_{2N} A_{2N} \\
\vdots & \vdots & \ddots & \vdots \\
\gamma_{N1} A_{N1} & \gamma_{N2} A_{N2} & \cdots & \gamma_{NN} A_{NN}
\end{bmatrix},
\]
where each $A_{ij}$ is an $M\times M$ permutation matrix as in \eqref{eq:perm-matrix-Z2m}.
If
\[
\Delta \;=\;
\frac{\displaystyle\prod_{1\le j<i\le N}(\alpha_i-\alpha_j)(\beta_j-\beta_i)}
     {\displaystyle\prod_{i=1}^{N}\prod_{j=1}^{N}(\alpha_i-\beta_j)}.
\]
Then $\det(\bfA)=\Delta^{M}$.

\begin{proof}
By Proposition~\ref{prop:perm}, the matrices in $\{A_{ij}\}$ commute pairwise because each $A_{ij}$ is of the form $\bfP^{(m)}_{\bfu}$ for some $\bfu\in\bbZ_2^m$.
Applying Theorem~\ref{thm:Laplace} specialized to characteristic two (or \eqref{eq:laplace-char2}) to the $k=N$ blocks gives
\[
\det(\bfA)
=\det\!\Bigg(\sum_{\pi\in S_N}\;\prod_{i=1}^N \gamma_{i,\pi(i)}\,A_{i,\pi(i)}\Bigg).
\]
For each $\pi\in S_N$, Proposition~\ref{prop:perm} implies
$\prod_{i=1}^N A_{i,\pi(i)}=\bfP^{(m)}_{u(\pi)}$ with
$u(\pi)=\sum_{i=1}^N u_{i,\pi(i)}\in\bbZ_2^m$, so
\[
\det(\bfA)
=\det\!\Bigg(\sum_{\pi\in S_N}\alpha_\pi\,\bfP^{(m)}_{u(\pi)}\Bigg),
\text{ where }
\alpha_\pi\triangleq \prod_{i=1}^N \gamma_{i,\pi(i)}.
\]
By Lemma~\ref{thm:perm_determinant}, the determinant of a linear combination of the permutation matrices $\bfP^{(m)}_\bfu$ depends only on the sum of the coefficients. Specifically, we have that 
\[
\det(\bfA)=\Big(\sum_{\pi\in S_N}\alpha_\pi\Big)^{M}.
\]
Finally, $\sum_{\pi\in S_N}\alpha_\pi=\det\!\big[\gamma_{ij}\big]_{i,j=1}^N$ and the Cauchy determinant evaluation~\cite{cauchy1841exercices} gives $\det[\gamma_{ij}]=\Delta$ 
Therefore $\det(\bfA)=\Delta^{M}$.
\end{proof}

\end{document}